\documentclass[11pt]{article}
\usepackage[english]{babel}
\usepackage{float}
\usepackage{lineno} 
\usepackage{graphicx,multicol}
\usepackage{epic,eepic,epsfig}
\usepackage{amssymb}
\usepackage{amsmath,amsthm}
\usepackage{color}
\usepackage{tikz}
\newcommand{\jbj}[1]{#1}
\newcommand{\JBJ}[1]{#1}

\usepackage{eso-pic} 

\newcommand{\AY}[1]{#1}

\newcommand{\induce}[2]{\mbox{$ #1 \langle #2 \rangle$}}

\newcommand{\2}{\vspace{2mm}}

\theoremstyle{plain}
\newtheorem{theorem}{Theorem}
\newtheorem{proposition}[theorem]{Proposition}
\newtheorem{claim}{Claim}[theorem]

\newtheorem{corollary}[theorem]{Corollary}
\newtheorem{lemma}[theorem]{Lemma}

\theoremstyle{definition}

\newtheorem{problem}[theorem]{Problem}
\newtheorem{conjecture}[theorem]{Conjecture}

\newtheorem{question}[theorem]{Question}

\DeclareMathOperator{\ind}{ind}

\begin{document}
\bibliographystyle{plain}

\title{Non-separating spanning trees and out-branchings in digraphs of independence number 2}
\author{J. Bang-Jensen\thanks{Department of Mathematics and Computer
    Science, University of Southern Denmark, Odense, Denmark (email:
    jbj@imada.sdu.dk). Part of this work was done while the author was visiting LIRMM, Universit\'e de
    Montpellier as well as  INRIA Sophia Antipolis. Hospitality and financial support by both is gratefully acknowledged. Ce travail a b\'en\'efici\'e d'une aide du gouvernement français, 
g\'er\'ee par l'Agence Nationale de la Recherche au titre du projet Investissements d'Avenir UCAJEDI portant la r\'ef\'erence no ANR-15-IDEX-01.}\and S. Bessy\thanks{LIRMM, CNRS, Universit\'e de
    Montpellier, Montpellier, France (email:stephane.bessy@lirmm.fr), financial suports: PICS CNRS DISCO and ANR DIGRAPHS n.194718.}\and  A. Yeo\thanks{Department of
    Mathematics and Computer Science, University of Southern Denmark,
    Odense, Denmark and  Department of Mathematics, University of Johannesburg, Auckland Park, 2006 South Africa (email: yeo@imada.sdu.dk).}}

\maketitle

\begin{abstract}
  A subgraph $H=(V,E')$ of a graph $G=(V,E)$ is {\bf non-separating} if $G\setminus{}E'$, that is, the graph obtained from $G$ by deleting the edges in $E'$,  is connected. 
  Analogously we say that a subdigraph $X=(V,A')$ of a digraph $D=(V,A)$ is non-separating if $D\setminus A'$ is strongly connected.
  We study non-separating spanning trees and out-branchings in digraphs of independence number 2. 
  Our main results are  that every 2-arc-strong digraph $D$ of independence number $\alpha{}(D)=2$ and minimum \JBJ{in-degree} at least \AY{5} 
  and every 2-arc-strong oriented  graph with $\alpha{}(D)=2$ and minimum \JBJ{in-degree} at least \AY{3 has a non-separating out-branching
  and minimum \JBJ{in-degree} 2 is not enough.}
  We also prove a number of other results, including that every \JBJ{2-arc-strong digraph $D$ with $\alpha{}(D)\leq 2$ and at least 14 vertices has a non-separating spanning tree and that every} graph $G$ with $\delta{}(G)\geq 4$ and $\alpha(G)=2$ has a non-separating hamiltonian path.\\
  \noindent{}{\bf Keywords:} non-separating branching; spanning trees; digraphs of independence number 2; strongly connected; hamiltonian path.
  
\end{abstract}

\section{Introduction}
\JBJ{An {\bf out-tree} in a digraph $D=(V,A)$ is a connected subdigraph $T^+_s$ of $D$ in which every vertex of $V(T^+_s)$, except one vertex $s$ (called the {\bf root}) has exactly one arc entering. This is equivalent to saying that $s$ can reach every other vertex of $V(T^+_s)$ by a directed path using only arcs of $T^+_s$.
An {\bf out-branching} in a digraph $D=(V,A)$ is a spanning  out-tree, that is, every vertex of $V$ is in the tree. We use the notation  
$B^+_s$ for an out-branching rooted at the vertex $s$.  An {\bf in-branching}, $B^-_t$, rooted at the vertex $t$ is defined analogously.}
The following classical result due to Edmonds and the algorithmic proof due to Lov\'asz \cite{lovaszJCT21} implies that one can check the existence of $k$ arc-disjoint out-branchings in polynomial time. 
\begin{theorem}[Edmonds]\cite{edmonds1973}
  \label{thm:Edbrthm}
  Let $D=(V,A)$ be a digraph and let $s\in V$. Then $D$ contains  $k$ arc-disjoint out-branchings, all rooted at $s$, if and only if there are $k$ arc-disjoint $(s,v)$-paths in $D$ for every $v\in V$.
\end{theorem}

Deciding the existence of arc-disjoint in-and out-branchings is considerably more difficult as shown by the following result due to Thomassen (the theorem and its proof can be found in \cite{bangJCT51}).

\begin{theorem}
  \label{thm:inoutNPC}
  It is NP-complete to decide whether a digraph contains arc-disjoint branchings $B^+_s,B^-_t$ for given vertices $s,t$.
\end{theorem}

It was shown in \cite{bangDAM161a} that the problem remains NP-complete even for 2-arc-strong 2-regular digraphs.

Thomassen conjectured that sufficiently high arc-connectivity will guarantee the existence of arc-disjoint in- and out-branchings with prescribed roots. \JBJ{As defined in Section \ref{sec:term}, $\lambda{}(D)$ denotes is the arc-connectivity of the digraph $D$.}

\begin{conjecture}\cite{thomassen1989}
  \label{conj:higharcBr}
  There exists a natural number $K$ such that every digraph $D$ with $\lambda{}(D)\geq K$
  contains arc-disjoint branchings $B^+_s,B^-_t$ for every choice of $s,t\in V$.
\end{conjecture}

It was pointed out in \cite{bangInOutB} that Conjecture \ref{conj:higharcBr} is equivalent to the following (the same value of $K$ would work for both conjectures).

\begin{conjecture}
  \label{conj:somepairofB}
  There exists a natural number $K$ such that every digraph $D$ with $\lambda{}(D)\geq K$
  contains an out-branching which is arc-disjoint from some in-branching.
\end{conjecture}

In this paper we study digraphs of independence number 2. The structure of digraphs with independence number 2 is  not  well understood and there are numerous interesting open problems. For instance it is an open problem whether the existence of vertex disjoint paths $P_1,P_2$ such that $P_i$ is an $(s_i,t_i)$-path for $i=1,2$ can be checked in polynomial time (for a partial result see \cite{chudnovskyJCT135}). In the case where we want arc-disjoint paths, a polynomial algorithm was given in \cite{fradkinJCT110}.

  Very recently the following result which settles Conjectures \ref{conj:higharcBr} and \ref{conj:somepairofB} for digraphs of independence
number 2 was obtained. The result is best possible in
terms of the arc-connectivity as there are infinitely many strong digraphs
with independence number 2 and arbitrarily high minimum in-and
out-degrees that have  no out-branching which is arc-disjoint from some in-branching \cite{bangInOutB}. 
  \begin{theorem}\cite{bangInOutB}
    \label{conj:inoutbralpha2}
  Every digraph $D=(V,A)$ with $\alpha{}(D)=2$ and $\lambda{}(D)\geq 2$ contains arc-disjoint branchings $B^+_s,B^-_t$ for some choice of $s,t\in V$.
\end{theorem}

\begin{conjecture}\label{conj:2asalpha2}\cite{bangInOutB}
  Every 2-arc-strong digraph $D=(V,A)$ with $\alpha{}(D)=2$ has a pair of arc-disjoint branchings
  $B_s^+,B_s^-$ for every choice of $s\in V$.
\end{conjecture}

\begin{conjecture}\cite{bangInOutB}
  Every 3-arc-strong digraph $D=(V,A)$ with $\alpha{}(D)=2$ has a pair of arc-disjoint branchings
  $B_s^+,B_t^-$ for every choice of $s,t\in V$.
\end{conjecture}

In the present paper we are interested in the existence an out-branching $B^+$   in a strongly connected digraph $D$  of independence number 2 such that the digraph $D\setminus{}A(B^+)$ that we obtain by deleting all arcs of $B^+$ is still strongly connected. Clearly if $D$ has such an out-branching, then it also has arc-disjoint in- and out-branchings $B^+_s,B^-_s$ from some vertex $s$ (namely the root of $B^+$).
The main result of the paper is the following which, besides being of interest in connection with Conjecture \ref{conj:twostrong} below, also provides support for  Conjecture \ref{conj:2asalpha2}. 

\begin{theorem}\label{thm:main}
Let $D$ be a $2$-arc-strong digraph with $\alpha(D)\leq 2$.
If either of the following statements hold then there exists an out-branching, $B^+$, in $D$,
such that $D\setminus{}A(B^+)$ is strongly connected.

\begin{description}
\item[(i):]  $\delta^-(D) \geq 5$, or
\item[(ii):]  $\delta^-(D) \geq 3$ and $D$ is \JBJ{an oriented graph (no cycles of length 2).}
\end{description}
\end{theorem}

Theorem~\ref{thm:main} (ii) is best possible in the sense that \JBJ{there exists a digraph $\tilde{D}$, with $\alpha{}(\tilde{D})=2$ and $\lambda{}(\tilde{D})\geq 2$ (and therefore
  $\delta^-{}(\tilde{D}) \geq 2$)} which does not contain a non-separating out-branching. See Figure~\ref{fig:nonsepT} and Proposition~\ref{prop:nonsepT}.\\

\JBJ{In Section \ref{sec:term} we provide some preliminary results. In particular, we show in} Proposition \ref{prop:no2colstrong} that there are infinitely many 2-arc-strong digraphs with independence number 2 and high in- and out-degrees that do not have an arc-partition into two spanning strong subdigraphs, implying that we cannot replace $B^+$ \JBJ{in Theorem \ref{thm:main}} by some spanning strong subdigraph. \JBJ{We also describe some structural results on semicomplete digraphs that will be used in later sections. Finally we prove a structural result on strong spanning subdigraphs with few arcs in digraphs with $\alpha{}(D)=2\leq\lambda{}(D)$.}\\

\JBJ{In Section \ref{sec:SDcase} we characterize semicomplete digraphs with non-separating out-branchings and prove a more general result which will be used in the proof of Theorem \ref{thm:main}.}\\

\JBJ{In Section \ref{sec:proof} we prove Theorem \ref{thm:main} and in Section \ref{sec:nonsepT} we study non-separating spanning trees in digraphs of independence number 2. The main result here is Theorem \ref{thm:nonsep14} which says that every 2-arc-strong digraph $D$ with $\alpha{}(D)=2$ and $n\geq 14$ vertices has a non-separating spanning tree. We conjecture that already $n\geq 9$ is enough, which would be best possible, and prove this in the case when $D$ has a hamiltonian cycle and no cycle of length 2. In Section \ref{sec:removeHP} we construct an infinite family of 2-arc-strong  digraphs with  $\alpha=2$ for which every hamiltonian path is separating and in Section \ref{sec:nonsepHPG} we show that for undirected graphs with independence number 2  a non-separating hamiltonian path always exists, provided the minimum degree is at least 4. Finally, in Section \ref{sec:remarks} we pose a number of open problems. }

  \section{Terminology and Preliminaries}\label{sec:term}
  Terminology not defined here or above is consistent with \cite{bang2009}.
  Let $D=(V,A)$ be a digraph. \JBJ{The {\bf underlying graph} of $D$ is the graph $UG(D)=(V,E)$ where $uv\in E$ if and only if there is at least one arc between $u$ and $v$ in $D$.}
  For a non-empty subset $X\subset V$ we denote by $d_D^+(X)$
(resp. $d_D^-(X)$) the number of arcs with tail (resp. head) in $X$ and head
(resp. tail) in $V\setminus{}X$. We call $d_D^+(X)$ (resp. $d_D^-(X)$) the {\bf out-degree}
(resp. {\bf in-degree}) of the set $X$. Note that $X$ may be just a
vertex. We will drop the subscript when the digraph is clear from the
context. We denote by $\delta^0(D)$ the minimum over all in- and
out-degrees of vertices of $D$. This is also called the minimum {\bf
  semi-degree} of a vertex in $D$.  The {\bf arc-connectivity} of $D$,
denoted by $\lambda{}(D)$, is the minimum out-degree of a proper subset
of $V$. A digraph is {\bf strongly connected} (or just {\bf strong}) if
$\lambda{}(D)\geq 1$. \JBJ{An arc $a$ of a strong digraph $D$ is a {\bf cut-arc} if $D\setminus\{a\}$ is not strong.}

\JBJ{When $X$ is a subset of the vertices of a digraph $D$, we denote by $D[X]$ the subdigraph {\bf induced} by $X$, that is, the vertex set of $D[X]$ is $X$ and the arc set consists of those arcs of $D$ which have both end vertices in $X$.}

The {\bf independence number}, denoted $\alpha{}(D)$, of a digraph $D=(V,A)$ is the size of a largest subset  $X\subseteq V$ such that the subdigraph of $D$ induced by $X$ has no arcs. 

  A {\bf strong component} of a digraph $D$ is a maximal (with
  respect to inclusion) subdigraph $D'$ which is strongly
  connected. The strong components of $D$ are vertex disjoint and
  their vertex sets form a partition of $V(D)$. If $D$ has more than
  one strong component, then we can order these as $D_1,\ldots{},D_k$
  such that there is no arc from a vertex in $V(D_j)$ to a vertex in
  $V(D_i)$ when $j>i$. A strong component $D_i$ is {\bf initial} ({\bf
    terminal}) if there is no arc of $D$ which enters (leaves)
  $V(D_i)$.

  The following result is well-known and easy to show.
  
\begin{proposition}
  \label{prop:outbranching}
  A digraph $D$ has an out-branching if and only it it has precisely one initial strong component. In that case every vertex of the initial strong component can be the root of an out-branching in $D$.
\end{proposition}

\JBJ{A digraph is {\bf semicomplete} if it has no pair of nonadjacent vertices. A {\bf tournament} is a semicomplete digraph with no directed cycle of length 2.
  A digraph $D=(V,A)$ is {\bf co-bipartite} is it has a vertex-partition $V_1,V_2$ such that $D[V_i]$ is semicomplete for $i\in [2]$.\\
  
  We shall make use of the following classical result due Camion. He formulated it only for tournaments but it is easy to see that it holds for semicomplete digraphs also. 
\begin{theorem}\cite{camionCRASP249}
    \label{thm:camion}
    Every strongly connected semicomplete digraph is hamiltonian.
    \end{theorem}}

  \subsection{Non-separating strong spanning subdigraphs}

  The following conjecture, which would clearly imply Conjecture \ref{conj:higharcBr}, has been verified for semicomplete digraphs (see Theorem \ref{thm:decomp2asSD} below).

\begin{conjecture}\cite{bangC24}
  \label{conj:twostrong}
  There exists a natural number $K$ such that every digraph $D$ with $\lambda{}(D)\geq K$
  contains arc-disjoint spanning strong subdigraphs $D_1,D_2$.
  \end{conjecture}

The \JBJ{infinite family of digraphs described below} shows that no condition on semi-degree is enough
to imply the conclusion of Conjecture~\ref{conj:twostrong} for
2-arc-strong digraphs, even for digraphs with independence number 2.
  
\begin{proposition}
  \label{prop:no2colstrong}
  For every natural number $K$ there are infinitely many 2-arc-strong digraphs $D$ with $\alpha{}(D)=2$ and  $\delta^0(D)\geq K$  that have no pair of arc-disjoint spanning strong subdigraphs.
\end{proposition}

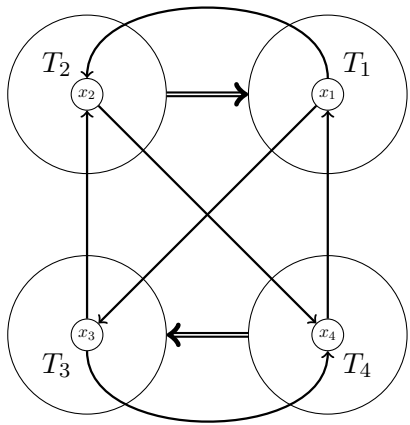
\begin{figure}[H] 
  \begin{center}
\tikzstyle{vertexL}=[circle,draw, minimum size=100pt, scale=0.6, inner sep=0.5pt]    
\tikzstyle{vertexB}=[circle,draw, minimum size=20pt, scale=0.6, inner sep=0.5pt]
\tikzstyle{vertexR}=[circle,draw, color=red!100, minimum size=14pt, scale=0.6, inner sep=0.5pt]

\begin{tikzpicture}[scale=0.8]
  \node (x1) at (5,5) [vertexB] {$x_1$};
  \node (t1) at (5,5)  [vertexL]{};
  \node (x2) at (1,5) [vertexB] {$x_2$};
  \node (t2) at (1,5)  [vertexL]{};
  \node (x3) at (1,1) [vertexB] {$x_3$};
  \node (t4) at (5,1)  [vertexL]{};
  \node (x4) at (5,1) [vertexB] {$x_4$};
  \node (t3) at (1,1)  [vertexL]{};
  \draw[->, line width=0.03cm] (x1) to (x3);
  \draw[->, line width=0.03cm] (x3) to (x2);
  \draw[->, line width=0.03cm] (x2) to (x4);
  \draw[->, line width=0.03cm] (x4) to (x1);
  \draw [->, thick, double] (t2) to (t1);
  \draw [->, thick, double] (t4) to (t3);
  \draw [->, line width=0.03cm] (x1) to  [in=90, out=90] (x2);
  \draw [->, line width=0.03cm] (x3) to  [in=-90, out=-90] (x4);
  \node  at (0.5,0.5) {$T_3$};
  \node at (5.5,0.5) {$T_4$};
  \node at (0.5,5.5) {$T_2$};
  \node at (5.5,5.5) {$T_1$};
   \end{tikzpicture}
   \caption{\JBJ{A 2-arc-strong digraph $D$ with $\alpha(D)=2$ and no decomposition into 2 arc-disjoint spanning subdigraphs}.}\label{fig:no2strong}
  
\end{center}
\end{figure}
\begin{proof}
Let $T$ be a 2-arc-strong tournament with $\delta^0(T)\geq K$ and let
$x$ be a vertex of $T$. Let $D=(V,A)$ be the digraph that we obtain
from 4 disjoint copies $T_i$, $i\in [4]$, of $T$ by adding the arcs of
the 4-cycle $x_1x_3x_2x_4x_1$, the arcs $x_1x_2,x_3x_4$, all possible
arcs from $V(T_2)$ to $V(T_1)$ and from \AY{$V(T_4)$ to $V(T_3)$.}
Here $x_i$ is the copy of $x$ in $T_i$. See Figure
\ref{fig:no2strong}.  Then $D$ is co-bipartite and 2-arc-strong and we
claim that $D'$ does not contain a pair of arc-disjoint spanning
strong subdigraphs.

Indeed, suppose there is a partition $A=A_1\cup A_2$ such that
$D_i=(V,A_i)$ is strong for $i=1,2$. There are exactly two arcs in
$D$ in both directions between $V(T_1)\cup V(T_2)$ and $V(T_3)\cup
V(T_4)$.  Without loss of generality we have $x_4x_1\in A_1$ and
$x_3x_2\in A_2$.  As there are only two arcs \AY{entering $V(T_2)$},
this implies that the arc $x_1x_2$ must be in $A_1$ (in order to reach
the vertices in $V(T_2)$) and \AY{as there are only two arcs leaving
  $V(T_1)$} we have $x_1x_3\in A_2$.  We must also have $x_3x_4\in
A_1$, since the only other arc leaving $V(T_3)$ is in $A_2$.  This
implies that the arc $x_2x_4$ must be in $A_2$ now we see that there
is no path from $V(T_1)\cup V(T_2)$ to $V(T_3)\cup V(T_4)$ in $D_1$,
contradiction.
\end{proof}

\subsection{Structure of semicomplete digraphs}\label{subsec:structure}

Let $D$ be a digraph.  A {\bf decomposition} of $D$ is a partition
$(S_1, \dots , S_p)$, $p\geq 1$, of its vertex set.  The {\bf index}
of vertex $v$ in the decomposition, denoted by $\ind(v)$, is the
integer $i$ such that $v\in S_i$.  An arc $uv$ is {\bf forward} if
$\ind(u) < \ind(v)$, {\bf backward} if $\ind(u) > \ind(v)$, and {\bf
  flat} if $\ind(u) = \ind(v)$.

A decomposition $(S_1,\ldots{},S_p)$ is {\bf strong} if $D\langle S_i
\rangle$ is strong for all $1\leq i\leq p$.  The following proposition
is well-known (just consider an acyclic ordering of the strong
components of $D$).
\begin{proposition}\label{prop:noback}
Every  digraph has a strong decomposition with no backward arcs.
\end{proposition}

A {\bf nice decomposition} of a strong digraph $D$ is a strong
decomposition such that the set of cut-arcs of $D$ is exactly the set
of backward arcs. \JBJ{Note that if $D$ has no cut-arc, that is, $\lambda(D)\geq 2$, then the strong decomposition with just one set $S_1=V(D)$ is nice.}

\begin{proposition}\cite{bangEulerSD}
  \label{prop:nice}
Every strong semicomplete digraph of order at least $4$ admits a nice
decomposition.
\end{proposition}

Given a semicomplete digraph and a nice decomposition of it, the {\bf
  natural ordering} of its backward arcs is the ordering \JBJ{of these arcs} in decreasing
order according to the index of their tail.  Note that this ordering
is unique \cite{bangEulerSD}. 

Denote by $S_4$ the semicomplete digraph on 4 vertices that we obtain
from a directed 4-cycle $v_0v_1v_2v_3v_0$ by adding the arcs
$v_0v_2,v_2v_0,v_1v_3,v_3v_1$. The following result shows that
Conjecture \ref{conj:twostrong} holds for semicomplete digraphs.

\begin{theorem}\cite{bangC24}
\label{thm:decomp2asSD}
  Let $D=(V,A)$ be a 2-arc-strong semicomplete digraph which is not
  isomorphic to $S_4$. Then $D$ contains two arc disjoint strong
  spanning subdigraphs $D_1,D_2$.
\end{theorem}

\begin{figure}[H]
\begin{center}
\tikzstyle{vertexX}=[circle,draw, top color=gray!5, bottom color=gray!30, minimum size=16pt, scale=0.6, inner sep=0.5pt]
\tikzstyle{vertexY}=[circle,draw, top color=gray!5, bottom color=gray!30, minimum size=20pt, scale=0.7, inner sep=1.5pt]
\begin{tikzpicture}[scale=0.35]
 \node (v1) at (1.0,6.0) [vertexX] {$d$};
 \node (v2) at (7.0,6.0) [vertexX] {$c$};
 \node (v3) at (1.0,1.0) [vertexX] {$a$};
 \node (v4) at (7.0,1.0) [vertexX] {$b$};
\draw [->, line width=0.03cm] (v1) -- (v2);
\draw [->, line width=0.03cm] (v2) -- (v3);
\draw [->, line width=0.03cm] (v3) -- (v4);
\draw [->, line width=0.03cm] (v4) -- (v1);
\draw [->, line width=0.03cm] (v1) to [out=285, in=75] (v3);
\draw [->, line width=0.03cm] (v3) to [out=105, in=255] (v1);
\draw [->, line width=0.03cm] (v3) to [out=135, in=225] (v1);
\draw [->, line width=0.03cm] (v2) to [out=285, in=75] (v4);
\draw [->, line width=0.03cm] (v4) to [out=105, in=255] (v2);
\node at (4,-0.4) {$S_{4,1}$};
\end{tikzpicture}\hspace{1.1cm}
\begin{tikzpicture}[scale=0.35]
 \node (v1) at (1.0,6.0) [vertexX] {$d$};
 \node (v2) at (7.0,6.0) [vertexX] {$c$};
 \node (v3) at (1.0,1.0) [vertexX] {$a$};
 \node (v4) at (7.0,1.0) [vertexX] {$b$};
\draw [->, line width=0.03cm] (v1) to [out=25, in=155] (v2);
\draw [->, line width=0.03cm] (v1) to [out=355, in=185]  (v2);
\draw [->, line width=0.03cm] (v2) -- (v3);
\draw [->, line width=0.03cm] (v3) -- (v4);
\draw [->, line width=0.03cm] (v4) -- (v1);
\draw [->, line width=0.03cm] (v1) to [out=285, in=75] (v3);
\draw [->, line width=0.03cm] (v3) to [out=105, in=255] (v1);
\draw [->, line width=0.03cm] (v2) to [out=285, in=75] (v4);
\draw [->, line width=0.03cm] (v4) to [out=105, in=255] (v2);
\node at (4,-0.4) {$S_{4,2}$};
\end{tikzpicture} \hspace{1.1cm}
\begin{tikzpicture}[scale=0.35]
 \node (v1) at (1.0,6.0) [vertexX] {$d$};
 \node (v2) at (7.0,6.0) [vertexX] {$c$};
 \node (v3) at (1.0,1.0) [vertexX] {$a$};
 \node (v4) at (7.0,1.0) [vertexX] {$b$};
\draw [->, line width=0.03cm] (v1) -- (v2);
\draw [->, line width=0.03cm] (v2) -- (v3);
\draw [->, line width=0.03cm] (v3) -- (v4);
\draw [->, line width=0.03cm] (v4) -- (v1);
\draw [->, line width=0.03cm] (v1) to [out=285, in=75] (v3);
\draw [->, line width=0.03cm] (v3) to [out=105, in=255] (v1);
\draw [->, line width=0.03cm] (v3) to [out=135, in=225] (v1);
\draw [->, line width=0.03cm] (v2) to [out=315, in=45] (v4);
\draw [->, line width=0.03cm] (v2) to [out=285, in=75] (v4);
\draw [->, line width=0.03cm] (v4) to [out=105, in=255] (v2);
\node at (4,-0.4) {$S_{4,3}$};
\end{tikzpicture}
\caption{The digraphs $S_{4,1}$, $S_{4,2}$, $S_{4,3}$}\label{fig:S4add}
\end{center}
\end{figure}
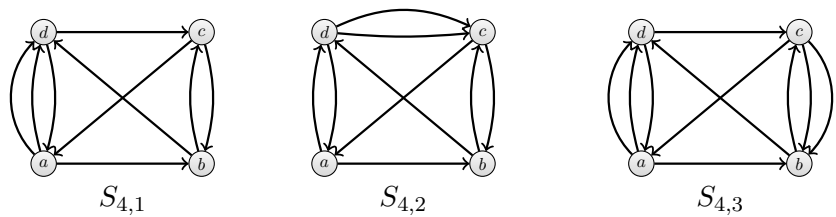

Recently Theorem \ref{thm:decomp2asSD} was extended to strong
decompositions of 2-arc-strong semicomplete directed multigraphs \JBJ{(parallel arcs allowed).}

  \begin{theorem}\cite{bangJGTSD}
    \label{thm:SDmultidecomp}
    Let $D$ be a 2-arc-strong semicomplete directed multigraph. Then
    $D$ has a pair of arc-disjoint strong spanning subdigraphs if and
    only if $D$ is not isomorphic to $S_4$ or one of three directed
    multigraphs shown in Figure \ref{fig:S4add} that can be obtained
    from $S_4$ by adding one or two extra arcs parallel to existing
    ones. Furthermore, if $D$ is not one of those four digraphs, then
    we can find a pair of arc-disjoint strong spanning subdigraphs in
    polynomial time.
  \end{theorem}

  \subsection{Strong spanning subdigraphs with few arcs in digraphs with $\alpha=2$.}\label{sec:strongspanning}

\begin{theorem}[Chen-Manalastras]\cite{chenDM44}
  \label{thm:CMthm}
  Let $D$ be a strongly connected digraph with $\alpha{}(D)=2$. Then
  either $D$ has a hamiltonian cycle or it has two cycles $C_1,C_2$
  that cover $V(D)$ and whose intersection is a (possibly empty)
  subpath of both cycles.
\end{theorem}

\begin{corollary}
  \label{cor:smallstrong}
  Let $D=(V,A)$ be a strong digraph with $\alpha{}(D)=2$.  Then (A) or
  (B) below holds.

  \begin{description}
  \item[(A):] $V(D)$ can be partitioned into $V_1$ and $V_2$, such
    that $D[V_i]$ are strong semicomplete digraphs for $i \in [2]$ and
    there exists $u_i \in V_i$ that is not adjacent to any vertex in
    $V_{3-i}$.

  \end{description}

  \begin{description}
    \item[(B):] $D$ has a strong spanning subdigraph $S$ with one of
      the following properties.

  \begin{description}
  \item[(B1)] $S$ is a hamiltonian cycle of $D$.

  \item[(B2)] There are two vertices $x,y$ of $S$ such that
    $d_S^+(x)=d_S^-(y)=1$, $d_S^-(x)=d_S^+(y)=2$ and
    $d_S^+(z)=d_S^-(z)=1$ for all $z\in V-\{x,y\}$.

Furthermore \JBJ{$N_S^-(x)$ and $N_S^+(y)$} are independent sets in $D$.

  \item[(B3)] There exists a vertex $x\in V$ such that
    $d_S^+(x)=d_S^-(x)=2$ and $d_S^+(v)=d_S^-(v)=1$ for all $v\neq x$.

Furthermore $N_S^+(x)$ and $N_S^-(x)$ are independent sets in $D$.
    \end{description}
  \end{description}

    \noindent{}In particular when one of (B1)-(B3) holds, the sum of
    the degrees of any two distinct vertices of $S$ is at most 6.
  \end{corollary}
  \begin{proof}
    Let $D$ have $\alpha{}(D)=2$. By Theorem \ref{thm:CMthm} $D$ has a
    hamiltonian cycle or it has two cycles $C_1,C_2$ that cover $V(D)$
    and whose intersection is a (possibly empty) subpath of both
    cycles.  If $D$ has a hamiltonian cycle we take $S$ to be that
    cycle and we are done as (B1) holds.  So now assume that $D$
    contains no hamiltonian cycle, which by Theorem~\ref{thm:CMthm}
    implies that $D$ contains two cycles $C_1,C_2$ that cover $V(D)$
    and whose intersection is a (possibly empty) subpath of both
    cycles. Let such $C_1$ and $C_2$ be chosen such that $|V(C_1)\cap
    V(C_2)|$ is maximum possible.

We now consider the case when $|V(C_1)\cap V(C_2)|>0$. If $|V(C_1)\cap
V(C_2)|=1$, then let $x$ be the vertex in $V(C_1)\cap V(C_2)$ and let
$A(S)$ to be the union of $A(C_1)$ and $A(C_2)$. Now the first part of
(B3) holds.  If $N_S^+(x)$ is not an independent set, then without
loss of generality assume that $xu \in A(C_1)$ and $xv \in A(C_2)$ and
$uv \in A(D)$.  Now remove the arc $xv$ from $C_2$ and add the path
$xuv$, in order to obtain a new cycle $C_2'$, with $|V(C_1)\cap
V(C_2')|=2>1=|V(C_1)\cap V(C_2)|$, and thereby contradicting the
maximality of $|V(C_1)\cap V(C_2)|$.  Therefore $N_S^+(x)$ is an
independent set. We can analogously show that $N_S^-(x)$ is an
independent set, and therefore part (B3) holds.  This completes the
case when $|V(C_1)\cap V(C_2)|=1$.  We may therefore assume that
$|V(C_1)\cap V(C_2)|\geq 2$.  That is, there are vertices $x,y\in
V(C_1)\cap V(C_2)$ such that the path common to $C_1$ and $C_2$ is $P$
and $P=C_i[x,y]$ for $i=1,2$. Now the first part of (B2) holds. If
$N_S^+(y)$ is not an independent set then analogously to above we get
a contradiction to the maximality of $|V(C_1)\cap V(C_2)|$ (or to $D$
not being hamiltonian).  And, again analogously to above, we can show
that $N_S^-(x)$ is also an independent set.  Therefore (B2) holds in
this case.  This completes the case when $|V(C_1)\cap V(C_2)|>0$.

Now assume that $|V(C_1)\cap V(C_2)|=0$ and therefore $C_1$ and $C_2$
are vertex disjoint.  As $D$ is strongly connected there exists a
$(C_1,C_2)$-arc, say $x_1 x_2 \in A(D)$.  Let $x_1^+$ be the successor
of $x_1$ on $C_1$. If there is any $(C_2,x_1^+)$-arc, $y_2 x_1^+$, in
$D$, then considering the cycle $C_1[x_1^+,x_1] C_2[x_2,y_2] x_1^+$
instead of $C_1$, would contradict the maximality of $|V(C_1)\cap
V(C_2)|$. So there is no $(C_2,x_1^+)$-arc in $D$. If there is an
$(x_1^+,C_2)$-arc in $D$, then consider $x_1^+$ instead of $x_1$.
Continuing this process either gives us a vertex which is not adjacent
to any vertex in $C_2$ or there is no arc from $C_2$ to $x_1$,
$x_1^+$, $\left(x_1^+\right)^+$, etc., a contradiction to $D$ being
strong. So there must be a vertex $u_1 \in V(C_1)$ which is not
adjacent to any vertex in $C_2$.

Analogously we can show that there must be a vertex $u_2 \in V(C_2)$
which is not adjacent to any vertex in $C_1$. This implies that
$D[V(C_i)]$ is semicomplete, as if two vertices, $x_i,y_i$, in
$D[V(C_i)]$ are non-adjacent then $\{x_i,y_i,u_{3-i}\}$ is an
independent set, a contradiction to $\alpha(D)=2$.

\section{Non-separating out-branchings in semicomplete digraphs}\label{sec:SDcase}

  \begin{figure}[H]
\begin{center}
  \tikzstyle{vertexB}=[circle,draw, minimum size=18pt, scale=0.7, inner sep=0.5pt]
\tikzstyle{vertexR}=[circle,draw, color=red!100, minimum size=14pt, scale=0.7, inner sep=0.5pt]
\begin{tikzpicture}[scale=0.6]
  \node (o) at (0,2) [vertexB] {$p_2$};
  \node (u) at (2,0) [vertexB] {$p_3$};
  \node (s) at (2,4) [vertexB] {$p_1$};
  \node (i) at (4,2) [vertexB] {$p_4$};
  \draw[->, line width=0.03cm] (s) to (o);
  \draw[->, line width=0.03cm] (s) to (u);
  \draw[->, line width=0.03cm] (i) to (s);
  \draw[->, line width=0.03cm] (i) to (o);
  \draw[->, line width=0.03cm] (o) to (u);
  \draw[->, line width=0.03cm] (u) to (i);
  \draw[->, line width=0.03cm] (o) to [out=-20,in=200] (i);
  \node () at (2,-1) {$W_1$};

  \node (r1) at (8,2) [vertexB] {$r_1$};
  \node (t1) at (10,0) [vertexB] {$t_1$};
  \node (r2) at (10,4) [vertexB] {$r_2$};
  \node (t2) at (12,2) [vertexB] {$t_2$};
  \draw[->, line width=0.03cm] (r1) to (t2);
  \draw[->, line width=0.03cm] (r2) to (t1);
  \draw[->, line width=0.03cm] (r1) to (t1);
  \draw[->, line width=0.03cm] (t1) to (t2);
  \draw[->, line width=0.03cm] (t2) to (r2);
  \draw[->, line width=0.03cm] (r2) to (r1);
  \node () at (10,-1) {$W_2$};
  
  \end{tikzpicture}
\end{center}
\caption{The semicomplete digraphs $W_1$ and $W_2$.}\label{fig:SDnononsepB}
\end{figure}
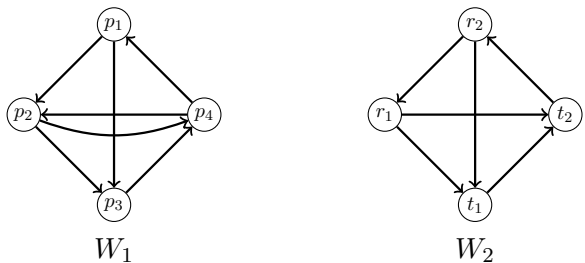

\begin{theorem}  \label{thm:SDnonsepB}
Let $D$ be a strong semicomplete digraph. Then the following holds. \\

  (a) If $D$ has at least two vertices with in-degree one, then $D$
contains no non-separating branching.  Furthermore if $D$ contains
exactly two vertices with in-degree one and is not isomorphic to $W_2$
(see Figure~\ref{fig:SDnononsepB}), then there exists an out-tree
$T^+$ rooted at $r_1$, such that $V(T^+)=V(D)-r_2$ and $D\setminus{}A(T^+)$ is
strong, where $d_D^-(r_1)=d_D^-(r_2)=1$.\\
  
  (b) If $D$ is isomorphic to $W_1$ (see
Figure~\ref{fig:SDnononsepB}), then $D$ contains no non-separating
branching.\\

  (c) If $D$ is not isomorphic to $W_1$ and contains exactly one
vertex, $r$, of in-degree one, then $D$ contains a non-separating
branching, rooted at $r$. \\

  (d) If $\delta^-(D)\geq 2$ and $|V(D)|\leq 3$, then for every $r \in
V(D)$ the digraph $D$ contains a non-separating branching, rooted at
$r$.  \\

  (e) If $\delta^-(D)\geq 2$ and $|V(D)|\geq 4$, then $D$ admits a
nice decomposition $(S_1,S_2,\ldots,S_p)$, and for every $r \in S_1$
the digraph $D$ contains a non-separating branching, rooted at $r$.
 \end{theorem}

\begin{proof}
 Recall that in an out-branching every vertex except the root has
  one arc entering it. Hence if a vertex has in-degree one in $D$, it
  must be the root of any non-separating out-branching. This shows
  that if $D$ admits a non-separating out-branching, it has at most
  one vertex with in-degree one. This proves the first part of (a).

Now assume that $D$ contains exactly two vertices, $r_1$ and $r_2$,
with in-degree one and let $H$ be a hamiltonian  cycle in $D$ \JBJ{($H$ exists by Theorem \ref{thm:camion})} and let
$D'=D\setminus{}A(H)$. If there is only one initial strong component in
$D'-r_2$, then letting $T^+$ be an out-branching in $D'-r_2$ gives us
the desired out-tree.  So assume that there are at least two initial
strong components in $D'-r_2$, $\{r_1\}$ and $S_1$. If $|V(S_1)| \geq
2$ then as $r_1$ is non-adjacent to $\{r_2\} \cup V(S_1)$  in $D'$
and $|\{r_2\} \cup V(S_1)| \geq 3$ we obtain a contradiction (as $H$
was a hamiltonian  cycle\JBJ{, meaning that we removed only 2 arcs incident to $r_1$ when we obtained $D'$ from $D$}).  So $|S_1|=1$ and we let $V(S_1)=\{t_1\}$. As
$d_D^-(t_1) \geq 2$, we have $d_{D'}^-(t_1) =1$ and $N_{D'}^-(t_1) =
\{r_2\}$.  Analogously considering $D'-r_1$ instead of $D'-r_2$ we
obtain an initial strong component $S_2$ in $D'-r_1$ where
$V(S_2)=\{t_2\}$ and $N_{D'}^-(t_2) = \{r_1\}$. Note that $t_1 \not=
t_2$ (as $r_1t_2,r_2t_1 \in A(D')$).  Furthermore $t_1$ and $t_2$ are
not adjacent in $D'$, as if $t_1 t_2 \in A(D')$ then $r_2 t_1 t_2$ is
a path in $D'-r_1$ and $S_2$ is not an initial strong component in
$D'-r_1$. So in $D'$, $r_1$ is non-adjacent to $\{r_2,t_1\}$ and
$t_2$ is non-adjacent to $\{r_2,t_1\}$. Therefore $|V(D)|=4$ and $D$
is isomorphic to $W_2$.  This proves the second part of (a).

  It is easy to check that if $D$ is
  the semicomplete digraph $W_1$ in Figure \ref{fig:SDnononsepB}, then
  every out-branching is separating (the vertex $p_1$ with in-degree one
  must be the root of all out-branchings), which proves part (b).

Now suppose that $D=(V,A)$ is different from $W_1$ and has exactly one
vertex of in-degree one, which implies that $n=|V(D)| \geq 3$.  Let
$H=p_1 p_2 \ldots p_n p_1$ be a hamiltonian  cycle in $D$ and let
$D'=D\setminus{}A(H)$.  Without loss of generality assume that $d_D^-(p_1)=1$,
which implies that $d_{D'}^-(p_1)=0$.  If $D'$ only has one initial
strong component, then $D'$ contains an out-branching, $B_{p_1}^+$,
rooted at $p_1$, which implies that $B_{p_1}^+$ is a non-separating
out-branching in $D$. Therefore we may assume that $D'$ contains at
least two initial strong components, one of which is just the vertex
$p_1$. As $d_{D'}^-(x) \geq 1$ for all $x \in V(D') \setminus \{p_1\}$
we note that any other initial strong component, $S$, in $D'$ must
contain at least two vertices.  Furthermore as there is no arc between
$p_1$ and the vertices in $S$ in $D'$, we must have $V(S)=\{p_2,p_n\}$
and $p_2 p_n, p_n p_2 \in A(D')$.  Therefore $n \geq 4$, as otherwise
$p_2 p_3 \in A(H)$ and $p_2 p_3 \in A(D')$, a contradiction.

If $n=4$, then we note that $D=W_1$, a contradiction (as
$H=p_1p_2p_3p_4p_1$ and $A(D') = \{p_2p_4,p_4p_2,p_1p_3\}$).  So
assume that $n \geq 5$. Let $T$ be obtained from $H$ by deleting the
arc $p_1p_2$ and adding the arcs $p_np_2$ and $p_1 p_3$. Note that $T$
is a strongly connected spanning subgraph of $D$.  Let $D^*=D\setminus{}A(T)$,
and note that $p_1 p_i \in A(D^*)$ for all $i \in [n] \setminus
\{n,3\}$ and $p_1 p_2, p_2 p_n, p_n p_3 \in A(D^*)$ (as the vertex set
of the initial component, $S$, in $D'$ was $\{p_2,p_n\}$).  Therefore
the only initial component in $D^*$ is $\{p_1\}$ and there exists an
out-branching, $B_{p_1}^+$ in $D^*$ rooted at $p_1$, which is
therefore a non-separating branching in $D$. This proves part (c).

We now consider the case when $\delta^-(D)\geq 2$. If $n \leq 3$, then
$D$ is the complete digraph on three vertices and part (d) holds. So
assume that $n \geq 4$.  By Proposition~\ref{prop:nice}, $D$ admits a
nice decomposition $(S_1,S_2,\ldots,S_p)$.

First consider the case when $p=1$. That is $D$ is $2$-arc-strong.  If
$D$ is isomorphic to $S_4$ (see Theorem \ref{thm:decomp2asSD}), then
as can be seen in Figure \ref{fig:S4OK}, $S_4$ has a non-separating
branching $B^+_{r}$ for each $r \in V(S_4)$.  So we may assume that
$D$ is not isomorphic to $S_4$, which by Theorem \ref{thm:decomp2asSD}
implies that $D$ contains two arc disjoint strong spanning subdigraphs
$D_1$ and $D_2$. For every $r \in V(D)$ we note that $D_1$ contains an
out-branching rooted at $r$ and therefore $D$ contains a
non-separating branching, rooted at $r$.  This proves part (e) when
$p=1$.

\begin{figure}[H]
\begin{center}
  \tikzstyle{vertexB}=[circle,draw, minimum size=18pt, scale=0.7, inner sep=0.5pt]
\tikzstyle{vertexR}=[circle,draw, color=red!100, minimum size=14pt, scale=0.7, inner sep=0.5pt]
\begin{tikzpicture}[scale=0.7]
  \node (a) at (0,0) [vertexB] {$a$};
  \node (b) at (4,0) [vertexB] {$b$};
  \node (c) at (4,4) [vertexB] {$c$};
  \node (d) at (0,4) [vertexB] {$d$};
  \draw[->, line width=0.03cm, color=blue] (a) to (b);
  \draw[->, line width=0.03cm,color=red] (b) to (d);
  \draw[->, line width=0.03cm,color=blue] (d) to (c);
  \draw[->, line width=0.03cm,color=blue] (c) to (a);
  \draw[->, line width=0.03cm,color=blue] (a) to [out=100, in=260] (d);
  \draw[->, line width=0.03cm,color=red] (d) to [out=-80, in=80] (a);
  \draw[->, line width=0.03cm,color=red] (c) to [out=-80, in=80] (b);
  \draw[->, line width=0.03cm,color=blue] (b) to [out=100, in=260] (c);

\end{tikzpicture}
\end{center}
\caption{Decomposing $S_4$ into an out-branching $B^+_c$ in red and a strong spanning subdigraph in blue.}\label{fig:S4OK}
\end{figure}
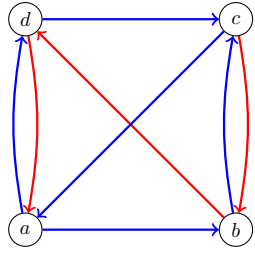
    
We now consider the case when $p \geq 2$. Let $r \in S_1$ be
arbitrary.  Let $st$ be the $(V(D)\setminus V(S_1),V(S_1))$-arc in
$D$. That is, $st$, is the cut-arc entering $S_1$.

Construct a new digraph $H_r$ from $S_1$ by adding a vertex $x$ to
$S_1$ and adding the two arcs $x t$ and $x r$ (if $t=r$ we add two
parallel arcs).  We will first show that $x$ has two arc-disjoint
paths to every other vertex in $H_r$.  As $S_1$ is strong we note that
$x$ can reach all other vertices in $H_r$ if we delete $x t$ or $ x
r$.  Furthermore if we delete any arc $e \in A(S_1)$ then $x$ can
still reach all other vertices in $S_1$ by starting with the arc $x
t$, as $t$ can reach all other vertices in $S_1$ even after deleting
one arc (by the definition of a nice decomposition).  By
Theorem~\ref{thm:Edbrthm} this implies that there exists two
arc-disjoint out-branchings both rooted in $x$ in $H_r$. Deleting $x$
from these gives us two arc-disjoint out-branchings $B_{t}^+, B^+_r$
in $S_1$, rooted at $t$ and $r$, respectively.

We will now show that we may assume that $B_t^+$ is not just an
out-star rooted at $t$. Assume that $B_t^+$ is an out-star rooted at
$t$. We first consider the case when $|S_1| \geq 4$. Let
$l_1,l_2,\ldots,l_{|S_1|-1}$ be the leaves of $B_t^+$ and note that
$\{l_1,l_2,\ldots,l_{|S_1|-1}\}$ is not independent in $D\setminus{}A(B_r^+)$ as
the underlying graph of $B_r^+$ is acyclic. So without loss of
generality we may assume that $l_1l_2 \in D\setminus{}A(B_r^+)$.  Now delete the
arc $tl_2$ from $B_t^+$ and add the arc $l_1l_2$ instead. We have
then obtained a $B_t^+$ (arc disjoint from $B_r^+$) that is not an
out-star, as desired. So we may now consider the case when $|S_1| \leq
3$. Notice that $|S_1|=1$ is impossible as $\delta^-(D)\geq 2$, so we
have $|S_1| \geq 2$. However if $|S_1|=2$, denoting $S_2$ by $\{t,y\}$,
we have $d_D^-(y)=1$, a contradiction. So we must have $|S_1|=3$.  Let
$S_1 = \{t,x,y\}$. As $d_D^-(y), d_D^-(x) \geq 2$ we note that
$xy,yx,tx,ty \in A(D)$. As $S_1$ is strong we note that $xt \in A(D)$
or $yt \in A(D)$ (or both). Without loss of generality assume that $xt
\in A(D)$. We now obtain the desired $B_t^+$ and $B_r^+$ as follows.

\begin{itemize}
  \item If $r=t$, then $B_t^+ = \{tx,xy\}$ and $B_r^+=\{ty,yx\}$.
  \item If $r=x$, then $B_t^+ = \{ty,yx\}$ and $B_r^+=\{xt,xy\}$.
  \item If $r=y$, then $B_t^+ = \{tx,xy\}$ and $B_r^+=\{yx,xt\}$.
\end{itemize}

Now, as $B_t^+$ is not just an out-star it contains a vertex $q$ which
is neither the root or leaf.  As $D$ is a strong semicomplete digraph
it contains a hamiltonian cycle, $H=p_1 p_2 p_3 \ldots p_n p_1$.  Without
loss of generality assume that the cut-arc $s t$ is the arc $p_n p_1$.
Then there must be exactly one arc in $H$ leaving $S_1$, say $p_i
p_{i+1}$. Note that $p_1 p_2 \ldots p_i$ is a hamiltonian path in
$S_1$ and $p_{i+1} p_{i+2} \ldots p_n$ is a hamiltonian path in $D -
V(S_1)$.  Let $Q$ be the union of $B_t^+$ and the path $p_{i+1}
p_{i+2} \ldots p_n p_1$ where we add an arc from every leaf of
$B_{t}^+$ to $p_{i+1}$ (which exists by the definition of the nice
decomposition and the fact that $t$ is not a leaf of $B_{t}^+$).  Note
that $Q$ is a strong spanning subdigraph of $D$.

Now construct $B^+$ by starting with $B_r^+$ and adding an arc from
$q$ (the vertex that was not the root or a leaf of $B_{t}^+$) to every
vertex in $V(D) \setminus V(S_1)$. Note that $B^+$ is an out-branching
in $D$ rooted at $r$ and $D - A(B^+)$ contains all arcs of $Q$ and is
therefore strongly connected.  This completes the proof of part~(e)
and therefore also of the theorem.
\end{proof}

As the digraph $S_4$ has a non-separating out-branching $B^+_v$ for
each of its 4 vertices, the same holds for any digraph obtained from
$S_4$ by adding arcs parallel to existing ones. Thus we one can prove
the following corollary of Theorem \ref{thm:SDmultidecomp}.

  \begin{corollary}
    \label{cor:SDmultiOK}
    Every 2-arc-strong semicomplete directed multigraph $D=(V,A)$ has
    a non-separating out-branching $B^+_v$ for every choice of $v\in
    V$.
    \end{corollary}

  \end{proof}


  \section{Proof of Theorem \ref{thm:main}}\label{sec:proof}

  Before we prove Theorem \ref{thm:main} we need the following lemma.
  
  \begin{lemma}
    \label{lem:4cycle}
    Let $D$ have $\alpha{}(D)=2\leq \lambda{}(D)$ and assume that
    $\delta^-(D) \geq 3$.  If $D$ satisfies (A) in Corollary
    \ref{cor:smallstrong}, then $D$ has a non-separating
    out-branching.
  \end{lemma}

  \begin{proof}
    Let $D$ be a digraph with $\alpha{}(D)=2 \leq \lambda{}(D)$ which
    consists of vertex disjoint strong semicomplete digraphs $D_1,
    D_2$, such that there exists $u_i \in V(D_i)$ that is not adjacent
    to any vertex in $D_{3-i}$ for $i=1,2$. 
    As
    $\delta^-(D) \geq 3$ we note that $d_{D_i}^-(u_i) \geq 3$.
    Therefore $|V(D_1)|,|V(D_2)| \geq 4$ and neither $D_1$ nor $D_2$
    is isomorphic with $W_1$ or $W_2$ (see
    Figure~\ref{fig:SDnononsepB}). We will now construct in $D$ an
    out-branching, $B^+$, and a spanning strong subdigraph, $Q$, which
    are arc-disjoint, as follows.  To start this construction,
    consider the following three cases for $i=1,2$. \\

{\em Case 1. There are at least two vertices of in-degree one in $D_i$.}

As $|V(D_i)| \geq 4$ and \JBJ{$D_i$} is a strong semicomplete digraph, we note
that there are exactly two vertices, $r_1^i$ and $r_2^i$, of in-degree
one \JBJ{(there can be at most 3 vertices of in-degree one in a semicomplete digraph and if there were 3 such vertices in $D_i$, then it would not be strong)}.  By Theorem~\ref{thm:SDnonsepB} and the fact that $D_i$ is not
isomorphic to $W_2$, there exists an out-tree \JBJ{$T_{r_i}^+$ rooted at $r_1^i$ and spanning $V(D_i)\setminus \{r_2^i\}$}  such that \JBJ{$D_i\setminus A(T_{r_i}^+)$} is
strongly connected.  Add the arcs of \JBJ{$T_{r_i}^+$} to $B^+$ and add the arcs
of \JBJ{$D_i\setminus A(T_{r_i}^+)$} to $Q$.  As $\delta^-(D) \geq 3$ we note that
there exists at least two $(D_{3-i},r_1^i)$-arcs and at least two
$(D_{3-i},r_2^i)$-arcs in $D$. \\

{\em Case 2. There is exactly one vertex, $r^i$, of in-degree one in
  $D_i$.}

As $D_i$ is not isomorphic to $W_1$, Theorem~\ref{thm:SDnonsepB}
implies that there is a non-separating out-branching, $B_{r^i}^+$, in
$D_i$, rooted at $r^i$. In this case add the arcs of $B_{r^i}^+$ to
$B^+$ and the remaining arcs of $D_i$ to $Q$.  As $\delta^-(D) \geq 3$
we note that there exists at least two $(D_{3-i},r^i)$-arcs in $D$. \\

{\em Case 3. $\delta^-(D_i) \geq 2$.}

As $|V(D_i)| \geq 4$, then Theorem~\ref{thm:SDnonsepB} (e) implies
that $D_i$ admits a nice decomposition
$(S_1^i,S_2^i,\ldots,S_{p_i}^i)$, and for every $r' \in S_1^i$ the
digraph $D_i$ contains a non-separating branching, rooted at $r'$.  As
$\lambda{}(D) \geq 2$, there must be a $(D_{3-i},S_1^i)$-arc, $ur^i$,
in $D$.  Let $B_{r^i}^+$ be a non-separating branching, rooted at
$r^i$ in $D_i$.  Add the arcs of $B_{r^i}^+$ to $B^+$ and the
remaining arcs of $D_i$ to $Q$. \\

This completes our three cases. Note that $Q$ contains a strong
spanning subdigraph of $D_1$ and of $D_2$. Furthermore in cases 2 and
3, $B^+$ contains an out-branching of $D_i$ rooted at a vertex $r^i$,
such that there exists a $(D_{3-i},r^i)$-arc in $D$. In Case~1, $B^+$
consists of an out-tree, rooted at $r_1^i$, containing all vertices of
$D_i$ except $r_2^i$, such that both $r_1^i$ and $r_2^i$ have at least
two arcs into them from $D_{3-i}$.  We now consider the following
possibilities.

We were in Case~2 or 3 for both $D_1$ and $D_2$.  Add an arc from
$D_1$ to the root of the out-branching of $D_2$ to $B^+$. As
$\lambda{}(D) \geq 2$, we can add a further $(D_1,D_2)$-arc and a
$(D_2,D_1)$-arc to $Q$, in order for $B^+$ and $Q$ to fulfill the
desired properties.

We were in Case~2 or 3 for $D_1$ and Case~1 for $D_2$.  Add an arc
from $D_1$ to $r_1^2$ and to $r_2^2$. As there were at least two
$(D_1,r_2^2)$-arcs in $D$ we can add a further $(D_1,r_2^2)$-arc to
$Q$ and as $\lambda{}(D) \geq 2$, we can add a $(D_2,D_1)$-arc to $Q$.
Now $B^+$ and $Q$ fulfill the desired properties.

We were in Case~2 or 3 for $D_2$ and Case~1 for $D_1$. This case is
analogous to the previous case.

We were in Case~1 for both $D_1$ and $D_2$.  Add an arc from $V(D_1)
\setminus \{r_2^1\}$ to $r_1^2$ and to $r_2^2$ to $B^+$ (which is
possible as $r_1^2$ and $r_2^2$ have at least two arcs into them from
$D_1$).  Also add an arc from $V(D_2) \setminus \{r_2^2\}$ to $r_2^1$
to $B^+$. Now $B^+$ is an out-branching rooted at $r_1^1$ in $D$.  As
there are at least two arcs into $r_1^i$ and into $r_2^i$ from
$D_{3-i}$ we note that there are at least four $(D_2,D_1)$-arcs and at
least four $(D_1,D_2)$-arcs in $D$. We can therefore add a
$(D_2,D_1)$-arc and a $(D_1,D_2)$-arc to $Q$ such that $Q$ and $B^+$
are arc-disjoint.  Now $B^+$ and $Q$ fulfill the desired properties,
completing the proof of the theorem.
    \end{proof}

Let us recall Theorem~\ref{thm:main}. \\

\noindent{}{\bf Theorem~\ref{thm:main}.}
{\em Let $D$ be a $2$-arc-strong digraph with $\alpha(D)\leq 2$.
If either of the following statements hold then there exists an out-branching, $B^+$, in $D$, 
such that $D\setminus{}A(B^+)$ is strongly connected.

\begin{description}
\item[(i):]  $\delta^-(D) \geq 5$, or
\item[(ii):]  $\delta^-(D) \geq 3$ and $D$ is oriented (has no 2-cycle).
\end{description}
} 

\begin{proof}
  Let $D$ be a $2$-arc-strong digraph with $\alpha(D) \leq 2$.  By
  Theorem~\ref{thm:SDnonsepB} we may assume that $\alpha{}(D)=2$.  By
  Lemma~\ref{lem:4cycle}, we may assume that $D$ has a strong spanning
  subdigraph $H$ satisfying one of the conditions (B1)-(B3) in Case~B of
  Corollary~\ref{cor:smallstrong}.  Let $D'=D\setminus A(H)$.

If $D'$ has only one initial strong component, then, by
Proposition~\ref{prop:outbranching} there is an out-branching in $D'$
and the theorem is proved.  So we may assume that
$R_1',R_2',\ldots,R_t'$ are the initial strong components in $D'$ and
$t \geq 2$.  For all $i \in [t]$, let $R_i = \induce{D}{V(R_i')}$.  We
will now prove the following claims.

\2

{\bf Claim A:} {\em $t=2$ and $|V(R_1)|,|V(R_2)| \geq 5$.
  Furthermore, all in-degrees in $D'$ are at least two, except
  possibly for one vertex whose indegree is at least one.  That is,
  there exists $r \in V(D')$, such that $d_{D'}^-(r) \geq 1$ and
  $d_{D'}^-(x) \geq 2$ for all $x \in V(D')\setminus \{r\}$.

 In Case~(i) we actually have $\delta^-(D') \geq 3$.}

{\bf Proof of Claim~A:} First consider Case~(i) (when $\delta^-(D)
\geq 5$).  As $\Delta^-(H)\leq 2$ we note that $\delta^-(D')\geq 3$,
and therefore also $\delta_{D'}^-(R_1) \geq 3$.  This further implies
that $|V(R_i)|=|V(R_i')| \geq 4$, since $N_{D'}^-[x] \subseteq
V(R_i')$, for all $x\in V(R'_i)$ and $i \in [t]$.  Now noting that
there is at most one  vertex of $H$ whose in-degree is more than 1, we
see that $R'_i$ contains a vertex with at least 4 in-neighbours inside
$R'_i$ so $|V(R_i)|=|V(R_i')| \geq 5$ holds.

Now consider the Case~(ii).  In this case we note that all in-degrees
in $D'$ are at least two, except possibly for one vertex whose
indegree is at least one.  Let $n_i=|V(R_i)|$ and note that the number
of arcs in $R_i$ is at least $2n_i-1$ (as all arcs into a vertex in
$R_i$ belong to $R_i$).  As $D$ is oriented this implies that ${n_i
  \choose 2} \geq |E(V_i)| \geq 2n_i-1$.  As ${n_i \choose 2} <
2n_i-1$ if $n_i \in [4]$, we must have $n_i \geq 5$, which implies
that $|V(R_i)| \geq 5$ in all cases.

For the sake of contradiction assume that $t \geq 3$.  Let $x_1 \in
V(R_1)$ be an arbitrary vertex with $d^+_H(x_1)=d^-_H(x_1)=1$ and let
$N_1$ be the set of the two neighbours of $x_1$ in $H$.  As
$|V(R_2)|\geq 5$ there are at least 3 vertices in $V(R_2)$ that are
not in $N_1$.  By part~(B) in Corollary~\ref{cor:smallstrong} we note
that at most two vertices in $H$ have degree more than two, so we can
can choose a vertex $x_2 \in V(R_2) \setminus N_1$ such that
$d^+_H(x_2)=d^-_H(x_2)=1$.  Let $N_2$ be the set of the two neighbours
of $x_2$ in $H$.  Now there exists a vertex $x_3 \in V(R_3) \setminus
(N_1 \cup N_2)$ which implies that $\{x_1,x_2,x_3\}$ is an independent
set in $D$, contradicting $\alpha(D) \leq 2$.  Therefore $t=2$, which
completes the proof of Claim~A.

\2

{\bf Claim B:} {\em $R_1$ and $R_2$ are semicomplete digraphs.
Furthermore, for all  $z \in V(D)$ the following holds,

\[
|N_H^+(z) \cap V(R_1)|,|N_H^+(z) \cap V(R_2)|,|N_H^-(z) \cap V(R_1)|,|N_H^-(z) \cap V(R_2)| \leq 1
\]
}

\2

{\bf Proof of Claim~B:} For the sake of contradiction assume that
$x_1,y_1 \in V(R_1)$ and $x_1$ and $y_1$ are non-adjacent in $D$.  Let
$N_2 = (N_H^+(x_1) \cup N_H^-(x_1) \cup N_H^+(y_1) \cup N_H^-(y_1))
\cap V(R_2)$.  If $V(R_2) \not\subseteq N_2$, then let $z_2 \in
V(R_2)\setminus N_2$ and note that $\{x_1,y_1,z_2\}$ is independent in
$D$, contradicting that $\alpha(D) \leq 2$.  So, $V(R_2) \subseteq
N_2$, which implies that $|N_2| \geq |V(R_2)| \geq 5$, by Claim~A.  By
Corollary~\ref{cor:smallstrong}, we note that $d_H(x_1)+d_H(y_1) \leq
6$, which implies the following,

{\small
\[
6 \geq |N_H^+(x_1) \cap V(R_2)| + |N_H^-(x_1) \cap V(R_2)|  + |N_H^+(y_1) \cap V(R_2)|  + |N_H^-(y_1) \cap V(R_2)|  \geq |N_2| \geq 5
\]
}

First consider the case when $|N_H^+(x_1) \cap V(R_2)| \geq 2$. By the
construction of $H$ we note that $|N_H^+(x_1) \cap V(R_2)|=2$, so let
$N_H^+(x_1) \cap V(R_2)=\{x_2,y_2\}$. By
Corollary~\ref{cor:smallstrong}, $x_2$ and $y_2$ are non-adjacent in
$D$.  As $\alpha(D)=2$ we note that either $x_2$ or $y_2$ has to be
adjacent to $y_1$.  Without loss of generality assume that $y_2$ is
adjacent to $y_1$.  As $y_2$ is adjacent to both $x_1$ and $y_1$ in
$D$ we note that we must have $d_H(x_1)+d_H(y_1) = 6$\JBJ{, as $H$ satisfies one of (B2), (B3) in  Corollary~\ref{cor:smallstrong}} (and
$|N_2|=|V(R_2)|=5$).

If $|N_H^-(x_1) \cap V(R_2)| \geq 2$ or $|N_H^+(y_1) \cap V(R_2)| \geq
2$ or $|N_H^-(y_1) \cap V(R_2)| \geq 2$ then we analogously can show
that $d_H(x_1)+d_H(y_1) = 6$ and there exists two non-adjacent
vertices $x_2$ and $y_2$ in $R_2$.

If we had considered $\{x_2,y_2\}$ instead of $\{x_1,y_1\}$ then we
would analogously have obtain $d_H(x_2)+d_H(y_2) = 6$.  However, by
part~(B) in Corollary~\ref{cor:smallstrong} we note that it is not
possible to have vertex-disjoint sets, $\{x_1,y_1\}$ and
$\{x_2,y_2\}$, such that $d_H(x_1)+d_H(y_1) = 6 = d_H(x_2)+d_H(y_2)$.
Therefore $R_1$ is semicomplete.  Analogously we can show that $R_2$
is also a semicomplete digraphs.

Let $z \in V(D)$ be arbitrary.  For the sake of contradiction assume
that $|N_H^+(z) \cap V(R_1)| \geq 2$.  By
Corollary~\ref{cor:smallstrong} we must have $|N_H^+(z) \cap V(R_1)| =
2$, so let $N_H^+(z) \cap V(R_1) = \{x_1,y_1\}$.  By
Corollary~\ref{cor:smallstrong} $x_1$ and $y_1$ are non-adjacent in
$D$. This contradicts the fact that $R_1$ is semicomplete. Therefore
$|N_H^+(z) \cap V(R_1)| \leq 1$.  Analogously $|N_H^+(z) \cap
R_2|,|N_H^-(z) \cap R_1|,|N_H^-(z) \cap R_2| \leq 1$, which completes
the proof of Claim~B.

\2

{\bf Claim C:} {\em Let $y \in V(D) \setminus (V(R_1) \cup V(R_2))$ be arbitrary. If $y$ has at most one arc entering it
from $V(R_1)$ in $D$, then $y$ is adjacent to all vertices in $V(R_2)$ and furthermore has at least four arcs entering it from $V(R_2)$.

Analogously, if $y$ has at most one arc entering it
from $V(R_2)$ in $D$, then $y$ is adjacent to all vertices in $V(R_1)$ and furthermore has at least four arcs entering it from $V(R_1)$.

The above implies that every vertex in $V(D) \setminus (V(R_1) \cup V(R_2))$ has at least four arcs entering it from
$V(R_1) \cup V(R_2)$ in $D$.}

\2

{\bf Proof of Claim~C:} Assume that $y \in V(D) \setminus (V(R_1) \cup
V(R_2))$ and $y$ has at most one arc entering it from $V(R_1)$ in
$D$. For the sake of contradiction assume that there exists $r_2 \in
V(R_2)$ which is non-adjacent to $y$ (in $D$).  By Claim~B, we note
that $|N_H^+(r_2) \cap V(R_1)|,|N_H^-(r_2) \cap V(R_1)|,|N_H^+(y) \cap
V(R_1)|,|N_H^-(y) \cap V(R_1)| \leq 1$.  As $R_1'$ and $R_2'$ are
initial strong components in $D'$, and therefore all arcs between
$r_2$ and $R_1$ in $D$ belong to $H$, we note that $r_2$ is adjacent
to at most two vertices in $R_1$ (in $D$).  As $y$ has at most one arc
entering it from $V(R_1)$ in $D$ and at most one arc from $y$ to $R_1$
in $H$ (and therefore also in $D$, as $R_1'$ is an initial strong
components in $D'$), we note that $y$ is adjacent to at most two
vertices in $R_1$ (in $D$).  As $|V(R_1)| \geq 5$, by Claim~A, this
implies that there exists a $r_1 \in V(R_1)$ which is not adjacent to
$r_2$ or $y$ in $D$, contradicting $\alpha(D)=2$. Therefore $y$ is
adjacent to every vertex in $R_2$.

As $R_2'$ was an initial strong component in $D'$ and $y \not\in
V(R_2')$ we note that every arc from $y$ to $R_2$ in $D$ belongs to
$H$. By Claim~A and Claim~B, we note that $|N_H^+(y) \cap V(R_2)| \leq
1$ and $|V(R_2)| \geq 5$, which implies that there are at least four
arcs from $V(R_2)$ to $y$ in $D$.  This completes the first part of
the proof of Claim~C.  The second part is proven analogously (by
swapping the names of $R_1$ and $R_2$).

Let $x \in V(D) \setminus (V(R_1) \cup V(R_2))$ be arbitrary.  If $x$
has less than two arcs entering it from $R_i$ then it has four arcs
entering it from $R_{3-i}$ ($i \in [2]$). And if $x$ has at least two
arcs entering it from both $R_1$ and from $R_2$, then it also has at
least four arcs entering it from $V(R_1) \cup V(R_2)$.  This completes
the proof of Claim~C.

\2

{\bf Construction of $R_1^*$.}  Initially let $R_1^*=R_1$. Now for
every $u \in V(D) \setminus (V(R_1^*) \cup V(R_2))$ with at least one
arc into $R_1^*$ in $D$ and at most one arc from $R_2$ to $u$, add $u$
to $R_1^*$.  Continue this process until no further vertex can be
added.



\begin{figure}[H]
\begin{center}
\tikzstyle{vertexB}=[circle,draw, minimum size=14pt, scale=0.6, inner sep=0.5pt]
\tikzstyle{vertexR}=[circle,draw, color=red!100, minimum size=14pt, scale=0.6, inner sep=0.5pt]

\begin{tikzpicture}[scale=0.6]]
  \draw (2,5) rectangle (6,8);
  \node () at (4,6.5) {$R_1$};

  \draw (3,0) rectangle (13,2);
  \node () at (8,1) {$R_2$};

  \node (v1) at (7,7) [vertexB] {$v_1$};
  \node (v2) at (9,7) [vertexB] {$v_2$};
  \node (v3) at (12,7) [vertexB] {$v_l$};

  \draw[dotted] (1.7,4.7) rectangle (8,8.3);
  \draw[dotted] (1.4,4.4) rectangle (11,8.6);
  \draw (1.0,4.0) rectangle (13,9);
  \node () at (14,6) {$R_1^*$};
  \node () at (10,6) {$\cdots$};

  \draw[->, dotted, line width=0.03cm] (7,2) to (v1);
  \draw[->, dotted, line width=0.03cm] (9,2) to (v2);
  \draw[->, dotted, line width=0.03cm] (11,2) to (v3);

  \node () at (19,3) {At most one arc from $R_2$ to $v_i$ for each $i$};

  \draw[->, line width=0.02cm] (13,3) to (12,3);

  \draw[->, line width=0.03cm] (v1) to [out=150,in=0] (5.5,7.4);
  \draw[->, line width=0.03cm] (v2) to [out=150,in=0] (7.5,7.4);
  \draw[->, line width=0.03cm] (v3) to [out=150,in=0] (10.5,7.4);

  \draw[->, line width=0.02cm] (6,5.65) to (v1);
  \draw[->, line width=0.02cm] (6,5.5) to (v1);
  \draw[->, line width=0.02cm] (6,5.35) to (v1);
  \draw[->, line width=0.02cm] (6,5.2) to (v1);

  \draw[->, line width=0.02cm] (6,5.65) to (v2);
  \draw[->, line width=0.02cm] (6,5.5) to (v2);
  \draw[->, line width=0.02cm] (6,5.35) to (v2);
  \draw[->, line width=0.02cm] (6,5.2) to (v2);

  \draw[->, line width=0.02cm] (6,5.65) to (v3);
  \draw[->, line width=0.02cm] (6,5.5) to (v3);
  \draw[->, line width=0.02cm] (6,5.35) to (v3);
  \draw[->, line width=0.02cm] (6,5.2) to (v3);
  \end{tikzpicture}
  \caption{An illustration of the construction of $R_1^*$, where $R_1^*$ is constructed from $R_1$ by adding the 
vertices $v_1,v_2,\ldots,v_l$ in that order. Every $v_i$ has at least one arc into $R_1 \cup \{v_1,v_2,\ldots,v_{i-1}\}$.
By Claim~C there are at least four arcs from $R_1$ into $v_i$ for each $i \in [l]$.} \label{fig:Rstar}
\end{center}
\end{figure}
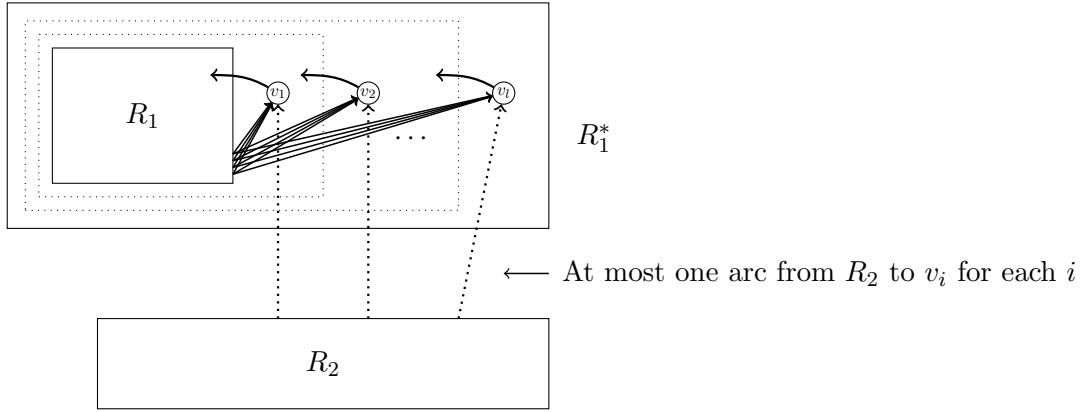

\2

{\bf Claim D:} {\em $R_1^*$ is a strong semicomplete digraph.}

\2

{\bf Proof of Claim~D:} Let $Q = V(R_1^*) \setminus V(R_1)$.  That is,
$Q$ denotes the set of vertices added to $R_1$ in the construction of
$R^*_1$. By construction, when it was added to the current $R^*_1$
each such vertex had at least one arc into the current set $V(R_1^*)$
and at most one arc entering it from $V(R_2)$ in $D$. We will first
show that $R_1^*$ is semicomplete.  Assume for the sake of
contradiction that $q_1,q_2 \in V(R_1^*)$ are non-adjacent in $D$.  By
Claim~B we note that $q_1$ and $q_2$ cannot both belong to $V(R_1)$.
By part 2 of Claim~C we note that we cannot have $q_1 \in V(R_1)$ and
$q_2 \in Q$ (or vice versa), which implies that we must have $q_1, q_2
\in Q$. This implies that $q_1$ and $q_2$ both have at most one arc
into them from $V(R_2)$.  By Claim~B, $|N_H^+(q_1) \cap
V(R_2)|,|N_H^+(q_2) \cap V(R_2)| \leq 1$, which implies that
$|N_D^+(q_1) \cap V(R_2)|,|N_D^+(q_2) \cap V(R_2)| \leq 1$.  Therefore
each of $q_1$ and $q_2$ are adjacent to at most two vertices in $R_2$.
As $|V(R_2)| \geq 5$, there is a vertex $r_2 \in V(R_2)$ which is
non-adjacent to both $q_1$ and $q_2$ in $D$, contradicting that
$\alpha(D) \leq 2$.  Therefore $R_1^*$ is semicomplete.
 
As $R_1$ is strongly connected and every vertex we add in the process
of building $R_1^*$ has an arc into the current set $R_1^*$ and an arc
(actually at least 4 arcs) entering it from $R_1$ (and $R_1 \subseteq
R_1^*$), by Claim~C, we note that the current set $R_1^*$ is strongly
connected in every step of the construction.  Therefore the final
$R_1^*$ is also strongly connected.  This completes the proof of
Claim~D.

\2

{\bf Definitions.} By Claim~A and D and Proposition~\ref{prop:nice} we
note that $R_1^*$ has a nice decomposition $(S_1, \dots , S_p)$.  If
$p \geq 2$ then there is only one arc entering $S_1$ in $R_1^*$ so let
$uv$ be an arc entering $S_1$ in $D$, which does not belong to
$R_1^*$. Such an arc exists as $D$ is $2$-arc-strong.  If $p=1$ then
$R_1^*$ is $2$-arc-strong. In this case let $uv$ be any arc entering
$R_1^*$ in $D$. Let $D^* = D' - uv$ (that is, delete the arc $uv$ from $D'$).

\2

{\bf Claim E:} {\em There exists an out-branching $B_1^+$ in $R_1^*$
  rooted at $v$, such that $R_1^* - A(B_1^+)$ is strongly connected.

Furthermore there exists an out-branching $B_2^+$ in $R_2$, such that 
$R_2 - A(B_2^+)$ is strongly connected. }

\2

{\bf Proof of Claim~E:} As $R_1^*$ is strongly connected by Claim~D,
We can apply Theorem~\ref{thm:SDnonsepB} to $R_1^*$.  By Claim~A we
note that $|V(R_1^*)| \geq 5$ and therefore $R_1^*$ is not isomorphic
to $W_1$.

Also, by Claim~A we note that all vertices of $R_1$, except possibly
one, have in-degree at least two in $R_1$.  As every vertex we add to
$R_1$ in the construction of $R_1^*$ have in-degree at least four
(from $R_1$, by Claim~C), we note that all vertices of $R_1^*$, except
possibly one, have in-degree at least two in $R_1^*$.  Therefore we
are in case (c)-(e) of Theorem~\ref{thm:SDnonsepB}.

Note that if there is a vertex of in-degree one in $R_1^*$, then $v$
is that vertex. Furthermore $v \in S_1$ (where $S_1$ was defined just
above Claim~E, as part of the nice decomposition on $R_1^*$).
Therefore by Theorem~\ref{thm:SDnonsepB} we note that $R_1^*$ contains
a non-separating out-branching rooted at $v$, which completes the
first part of the proof of Claim~E.

The second part of Claim~E, follows analogously, by Theorem~\ref{thm:SDnonsepB}.

\2

{\bf Completion of the proof.} Let $B^+= B_1^+ \cup B_2^+$ (defined in Claim~E) and let $Q = (R_1^* - A(B_1^+)) \cup (R_2 - A(B_2^+))$ 
and let $V^* = V(R_1^*) \cup V(R_2)$.
Note that $B^+$ consists of two vertex-disjoint out-trees (whose union span $V^*$) and $Q$ of two vertex-disjoint strong components 
(whose union also span $V^*$). 
Let $P_{12}$ be a path from $V(R_1^*)$ to $V(R_2)$ in $D^*$
and let $P_{21}$ be a path from $V(R_2)$ to $V(R_1^*)$ in $D^*$
Add the arcs of $P_{12}$ and $P_{21}$ to $Q$. For every vertex $p \in (V(P_{12}) \cup V(P_{21})) \setminus V^*$ do the following.
Add an arc, which is not in $A(P_{12}) \cup A(P_{21})$,  from $V(R_1^*) \cup V(R_2)$ to $p$ to $B^+$. 
This  is possible by Claim~C 
and the fact that there are at most two arcs in $A(P_{12}) \cup A(P_{21})$ leaving $V(R_1^*) \cup V(R_2)$ (at most one leaves $V(R_1^*)$ and 
at most one leaves $V(R_2)$).
Furthermore if $p = u$ (recall that the arc $uv$ was defined above Claim~E), then make sure the added arc leaves $V(R_2)$ (and not $V(R_1^*)$),
which is possible as $u$ was not added to $R_1^*$ and therefore has at least two arcs into it from $V(R_2)$ (here we used that the arc $uv$ enters $R^*_1$). 
Finally add $p$ to $V^*$.

When this process is completed $V^*=V(R^*_1)\cup V(R_2)\cup V(P_{12})\cup V(P_{21})$,
$Q$ induces a strong subgraph on the vertex set $V^*$ and $B^+$ still consist of two out-trees also spanning $V^*$, one of which is rooted at
$v$. Furthermore the arc $uv$ is not used above and if $u \in V^*$ then it belongs to the out-tree not rooted at $v$. We now add the remaining vertices
as follows.  While $V^* \not= V(D)$ let $P'$ be any $(V^*,V^*)$-path in $D^*$ with at least one internal vertex. We can construct $P'$ by letting 
$p_0p_1$ be any arc out of $V^*$ in $D^*$ and then taking any path from $p_1$ back to $V^*$ in $D^*$. Now add $A(P')$ to $Q$ and for every
vertex $p \in V(P') \setminus V^*$ do the following (analogously to above).
Add an arc, which is not in $A(P')$,  from $V(R_1) \cup V(R_2)$ to $p$ to $B^+$, which is possible by Claim~C and the 
fact that there is at most one arc in $A(P')$ leaving $V(R_1) \cup V(R_2)$.
Furthermore if $p = u$ (recall that the arc $uv$ was defined above Claim~E), then make sure the added arc leaves $V(R_2)$ (and not $V(R_1^*)$),
which is possible as $u$ was not added to $R_1^*$ and therefore has at least two arcs into it from $V(R_2)$ (here we used that the arc $uv$ enters $R^*_1$).
Finally add $p$ to $V^*$.

We continue the above process until $V^*=V(D)$. Now $Q$ is a strong spanning subgraph of $D$, which does not include any arcs from $B^+$ and also does
not include the arc $uv$. $B^+$ consists of two out-trees, one rooted at $v$ and $u$ belonging to the out-tree which was not rooted at $v$. Therefore by adding the 
arc $uv$ to $B^+$ we obtain a spanning out-branching of $D$ which is arc-disjoint to $Q$, thereby completing the proof.
\end{proof}


  \section{Non separating spanning trees in digraphs with independence number 2}\label{sec:nonsepT}

\jbj{Let us recall that for a subdigraph $H$ of $D$, we denote by
$D-H$ the sudigraph of $D$ obtained from removing the
vertices of $H$ from $D$, that is $D-H=D[V(D)\setminus V(H)]$}.
\AY{Furthermore we denote by $D\setminus A(H)$ the sudigraph of $D$ obtained from
removing the arcs of $H$ from $D$, that is $V(D-H)=V(D)$ and
$A(D-H)=A(D)\setminus A(H)$.}

A spanning tree $T$ of a connected digraph $D$ is {\bf safe} if for
every pair of distinct vertices $x$ and $y$ of $D$, there exists an oriented path from
$x$ to $y$ in $D$ if and only if there exists also an oriented path
from $x$ to $y$ in \AY{$D \setminus A(T)$}. In particular, a safe spanning tree of a
strong digraph is a non separating spanning tree.

 At several places,
we use the following fact. Assume that $H$ is an induced subdigraph of
$D$ such that $H$ admits a safe spanning tree $T$, and assume also
that there exist $u,v\in H$ and $x\in D\setminus H$ such that $ux,vx
\in A(D)$ and that there exist a path from $u$ to $v$ in $H$. Then
$D[V(H)\cup x]$ admits the safe spanning tree $T+ux$. Indeed, there
exists also a path from $u$ to $v$ in \AY{$H \setminus A(T)$} and thus a path from $u$
to $x$ in \AY{$D[V(H)\cup x]\setminus (A(T)\cup \{ux\})$}.

First, we \jbj{derive some results on} safe spanning trees of
semicomplete digraphs.

\begin{lemma}
\label{lem:safe-st}
Every semicomplete digraph on at least five vertices admits a safe
spanning tree.
\end{lemma}

\begin{proof}
Let $D$ be a semicomplete \AY{digraph} on at least five vertices. If $D$ is strong,
as $D$ contains at least five vertices, then we find a spanning tree
in the complement of a Hamiltonian cycle of $D$. This spanning tree is
clearly safe.

 So, assume that $D$ is not strong and denote by $C_1,
C_2, \dots ,C_t$ the strongly connected components of $D$ such that
there is no arc from $C_i$ to $C_j$ if $i>j$. We denote by $K$ the
subdigraph of $D$ containing the vertices $V(D)$ and the set of arcs
of $D$ connecting its strong components (that is the arcs $uv$ with
$u\in C_i$, $v\in C_j$ and $i\neq j$). Moreover, for every $i=1,\dots
,t$ let $x_i$ be a vertex of $C_i$ and $P$ be the path $x_1\dots
x_t$. If there exists a spanning tree $T$ of $D$ \jbj{all of whose
  arcs are in the subdigraph $K'=K\setminus A(P)$, then $T$ is a safe
  spanning tree of $D$}. Thus we have to check that \jbj{$K'$} is a
connected subdigraph of $D$.

 First, if there exist $i$ and $j$ such
that $|C_i|\ge 2$ and $|C_j|\ge 2$, then we pick a vertex $y_i$ in
$C_i$ different from $x_i$ and a vertex $y_j$ in $C_j$ different from
$x_j$. In \jbj{$K'$}, every vertex not in $C_i$ is adjacent to $y_i$
and every vertex in $C_i$ is adjacent to $y_j$. So, \jbj{$K'$} is
connected in this case.

 Moreover, if $t\ge 4$ then in \jbj{$K'$}
every vertex of $D\setminus C_1\cup C_2$ is adjacent to $x_1$ and
every vertex of $D\setminus C_{t-1}\cup C_t$ is adjacent to $x_t$. So,
\jbj{$K'$} is connected.

 Thus we may assume that $t\le 3$ and that
there exists $i_0\in \{1,\dots ,t\}$ such that for every $i\neq i_0$,
the component $C_i$ has size 1 exactly. If $t=3$, then, as $D$
contains at least 5 vertices, we have $|V(C_{i_0})|\ge 3$.  If
$i_0=1$, then \jbj{$K'$} contains all the arcs from \jbj{$V(C_1)-x_1$}
to $\{x_2,x_3\}$ and the arc $x_1x_3$. So \jbj{$K'$} is connected. The
case $i_0=3$ is symmetrical. 

\AY{So we may assume that $i_0=2$. Let $y_2 \in V(C_2) \setminus \{x_2\}$ 
be arbitrary and let $K^* = K \setminus \{x_1x_2,y_2x_3\}$. As every vertex 
in $V(C_2)$ is adjacent to $x_1$ or $x_3$ in $K^*$ and $x_1x_3 \in A(K^*)$ we
note that $K^*$ is a spanning connected subgraph of $D$ and we can therefore
in $K^*$ find a safe spanning tree of $D$.}

Finally, if $t=2$, by symmetry again we can assume that $i_0=1$. We have
$|V(C_1)|\ge 4$ and so $C_1$ contains an arc $xy$ such that
\jbj{$C_1\setminus xy$} is strongly connected. 
\AY{In this case} the
arcs from $V(C_1)\setminus \{x\}$ to $x_2$ plus the arc $xy$ form the
arcs of a safe spanning tree of $D$.
\end{proof}





 The following claim will be useful in the proof of the main theorem of
the section.

\begin{claim}
\label{claim:tree-extension}
Let $D$ be a digraph with $\lambda(D)\ge 2$. If $D$ contains a tree
$T$ such that \jbj{$D\setminus A(T)$} is strongly connected and
$V(D)\setminus V(T)$ has size at most two and \jbj{ induces a
  semicomplete digraph}, then $D$ admits a non-separating spanning
tree.
\end{claim}

\begin{proof}
Let $D$ and $T$ as stated, and let $C$ be a terminal strong
component of \jbj{$D[V(T)]\setminus A(T)$}.

\jbj{Suppose first that $V(D)\setminus V(T)=\{x\}$}. As $\lambda(D)\ge 2$, there are  at
least two arcs from $C$ to $x$. Let be $u$ an in-neighbour of $x$ in
$C$. The digraph \AY{$D \setminus (A(T) \cup \{ux\})$} is still strong and so $T+ux$ is a
non-separating spanning tree of $D$.

 Now, assume that \jbj{$V(D)\setminus V(T)=\{x,y\}$}.
 If there are  two
arcs $ux$ and $vx$ from $C$ to $x$, then $T'=T+ux$ is a tree on $n-1$
vertices such that \jbj{$D\setminus A(T')$} is strongly connected and  we can conclude from the  previous case that $D$ admits a separating spanning tree. So,
as $\lambda(D)\ge 2$, there are at least two arcs from $C$ to $\{x,y\}$
and we can assume that one, say $ux$, has head $x$ and the other, say
$vy$, has head $y$. Similarly, \jbj{we can assume by the previous case}, that if $C'$ is an initial strong component
of \jbj{$D[V(T)]\setminus A(T)$}, then there exist two arcs $xu'$ and $yv'$ with
$u',v'\in C'$. As \jbj{$D[\{x,y\}]$} is semicomplete, we can assume
without loss of generality that $xy$ is an arc of $D$. Now it is easy to  check
that \jbj{$D\setminus (A(T)\cup\{vy,xu'\})$} is strongly connected and that $D$ admits the
non-separating spanning tree $T+vy+xu'$. 
\end{proof}

Now we can prove the following.

\begin{theorem}
  \label{thm:nonsep14}
Every digraph $D=(V,A)$  with $\alpha(D)\le 2
\le \lambda(D)$ \jbj{such that $D$ contains a semicomplete digraph on at least 5 vertices has a non separating spanning tree. In particular, every digraph $D=(V,A)$  with $\alpha(D)\le 2
\le \lambda(D)$ such that $|V|\geq 14$ has a non-separating spanning tree.}
\end{theorem}

\begin{proof}
If \jbj{$D$ is semicomplete}, then the result follows from Lemma~\ref{lem:safe-st}.
So we may assume that $\alpha{}(D)=2$.  As the Ramsey number $R(3,5)$ is
14~\cite{radziszowskiEJC94} and $\alpha(D)=2$, \jbj{it follows that if $|V|\geq 14$, then $D$ contains a semicomplete subdigraph 
on five vertices. Hence we may assume below that $D_1$ is a semicomplete digraphs of $D$ on 5 vertices}. By Lemma~\ref{lem:safe-st},
$D_1$ contains a safe spanning tree. So let $R$ be a maximal induced
subdigraph of $D$ containing $V(D_1)$ and admitting a safe spanning tree.
\jbj{We now show that  $R=D$. Suppose for a contradiction that this is not the
case}.

 Let $T$ be a safe spanning
tree of $R$ and consider  a vertex $x$ of \jbj{$S=D[V\setminus V(R)]$}. The vertex
$x$ has at most one in-neighbour in $D_1$. Indeed, otherwise, assume
that $y$ and $z$ are two in-neighbours of $x$ in $D_1$ with \jbj{$yz$} being
an arc of $D_1$ (recall that $D_1$ is semicomplete). But then, $T+yx$
would be a safe spanning tree of $D[V(R)\cup x]$, a contradiction to
the maximality of $R$. Similarly, $x$ has at most one out-neighbour in
$D_1$. So we can conclude that $S$ is a semicomplete subdigraph of
$D$. Indeed, otherwise, $S$ would contain an independent set $\{u,v\}$
of size two. But as $u$ and $v$ have each at most two neighbours in
$D_1$, there would exist in $D_1$ a vertex not \jbj{adjacent  to any of} $u$ or $v$,
contradicting $\alpha(D)=2$.

First, assume that $S$ contains a safe spanning tree $T'$ and denote
by $C$ a strong terminal component of $R$.  As $\lambda (D)\ge 2$,
there exist at least two arcs $xu$ and $yv$ from $C$ to $S$ (with
$x,y\in C$ and $u,v\in S$). If $u\neq v$, then 
there is an arc between $u$ and $v$ as $S$ is semicomplete. \jbj{Without loss of generality}
assume that $uv$ is
an arc of $S$ and let $e$ be the arc $yv$. If $u=v$ then we can choose
arbitrarily $e=xu$ or $e=yv$. \jbj{In both cases }$T+T'+e$ is a safe spanning tree
of $D$, contradicting  the maximality of $R$.

So, $S$ has no safe spanning tree and as $S$ is semicomplete, it follows from 
Lemma~\ref{lem:safe-st}, that  $|V(S)|\le 4$. We also have
$|V(S)|>1$, as a unique vertex always has a safe spanning tree. Thus,
to conclude the proof of the Lemma, we have three cases to handle:
\jbj{$|V(S)|\in \{2,3,4\}$}.

 Assume first that $S$ contains
two vertices. Then, \AY{$D \setminus A(T)$} is strong, $D\setminus T$ has size two and
is semicomplete. So, by Claim~\ref{claim:tree-extension}, $D$ has a
separating spanning tree, a contradiction again to the maximality of
$R$.

 Now assume that $S$ contains three vertices, and denote by $C$
a strong terminal component of $R$.  As previously, as $\lambda (D)\ge
2$, there exist at least two arcs $xu$ and $yv$ from $C$ to $S$ (with
$x,y\in C$ and $u,v\in S$). If $u\neq v$, as $S$ is semicomplete
there is an arc between $u$ and $v$. Without loss of generality assume that $uv$ is
an arc of $S$ and let $e$ be the arc $yv$. If $u=v$ then we can choose
arbitrarily $e=xu$ or $e=yv$. To conclude, denote $T+e$ by $T'$ and
notice that \AY{$D \setminus A(T')$} is strongly connected. As $D\setminus T'$ has size
two and is semicomplete, by Claim~\ref{claim:tree-extension}, $D$ has
a separating spanning tree, a contradiction again to the maximality of
$R$.

 Finally, assume that $S$ contains four vertices. As in the
previous case, we can find an arc $e=zw$ with $z\in R$ and $w\in S$
such that \AY{$D \setminus (A(T) \cup \{e\})$} is strongly connected. As $S$ as four vertices,
$w$ has in-degree or out-degree at least 2 in $S$. Assume that $w$ has
out-degree at least 2 and denote by $u$ and $v$ two out-neighbours of
$w$ in $S$ such that $uv$ is an arc of $D$.  So, if we remove the arc
$wv$ from the digraph \AY{$D \setminus (A(T) \cup \{e\})$,} the resulting digraph still contains
the path $wuv$ from $w$ to $u$ and so is still strongly
connected. That is, if we denote by $T'$ the tree $T+e+wv$, the
digraph \AY{$D \setminus A(T')$} is strongly connected and \AY{$D - T'$} is
semicomplete and has size two. Thus by Claim~\ref{claim:tree-extension},
$D$ has a separating spanning tree, a contradiction again to the
maximality of $R$. \jbj{The case when $w$ has in-degree at least 2 in $S$ is analogous.}
\end{proof}

\begin{figure}[H]
\begin{center}
\tikzstyle{vertexB}=[circle,draw, minimum size=14pt, scale=0.6, inner sep=0.5pt]
\tikzstyle{vertexR}=[circle,draw, color=red!100, minimum size=14pt, scale=0.6, inner sep=0.5pt]

\begin{tikzpicture}[scale=1]
  \node (a) at (2,0) [vertexB] {$a$};
  \node (b) at (2,2) [vertexB] {$b$};
  \node (c) at (2,4) [vertexB] {$c$};
  \node (x) at (4,4) [vertexB] {$x$};
  \node (y) at (4,2) [vertexB] {$y$};
  \node (z) at (4,0) [vertexB] {$z$};
  \draw[->, line width=0.03cm] (a) to (b);
  \draw[->, line width=0.03cm] (b) to (c);
  \draw[->, line width=0.03cm] (x) to (y);
  \draw[->, line width=0.03cm] (y) to (z);
  \draw[->, line width=0.03cm] (a) to (y);
  \draw[->, line width=0.03cm] (y) to (c);
  \draw[->, line width=0.03cm] (c) to (x);
  \draw[->, line width=0.03cm] (x) to (b);
  \draw[->, line width=0.03cm] (b) to (z);
  \draw[->, line width=0.03cm] (z) to (a);
  \draw[->, line width=0.03cm] (c) to [out= 225,in=135] (a);
  \draw[->, line width=0.03cm] (z) to [out=45,in= -45] (x);

  \node (a) at (10,3) [vertexB] {$a$};
  \node (b) at (12,3) [vertexB] {$b$};
  \node (c) at (8,3) [vertexB] {$c$};
  \node (x) at (10,1) [vertexB] {$x$};
  \node (y) at (12,1) [vertexB] {$y$};
  \node (z) at (8,1) [vertexB] {$z$};
  \draw (7.5,0.5) rectangle (8.5,3.5);
  \draw (9.5,0.5) rectangle (10.5,3.5);
  \draw (11.5,0.5) rectangle (12.5,3.5);
  \draw[->, line width=0.03cm] (8.5,2) to (9.5,2);
  \draw[->, line width=0.03cm] (10.5,2) to (11.5,2);
  \draw[->, line width=0.03cm] (11.8,3.5) to [out=140,in= 40] (8.2,3.5);

  \end{tikzpicture}
  \caption{\AY{Two different drawings of the same} 2-arc-strong  co-bipartite digraph $\tilde{D}$ in which every spanning tree is
    separating.}\label{fig:nonsepT}
\end{center}
\end{figure}
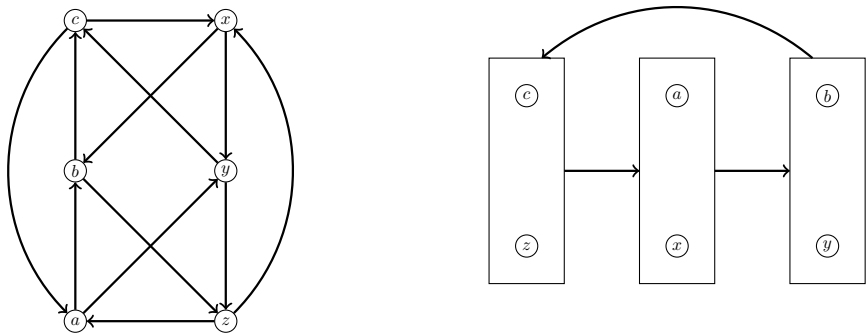

\begin{proposition}
  \label{prop:nonsepT}
  The digraph $\tilde{D}$ in Figure \ref{fig:nonsepT} has no non-separating spanning tree.
\end{proposition}
\begin{proof}
  Note that $H$ has 12 arcs and 6 vertices so if $H$ would have a pair of arc-disjoint subdigraphs $T,S$ where $T$ is a spanning tree and $S$ a strong spanning digraph, then $|A(S)|\leq 7$  must hold. This implies that $S$ is either a hamiltonian cycle of $H$ or it consist of a cycle $C$ and a $(C,C)$-path $P$ which picks up the remaining vertices of $V$.  Now note that the only cycle lengths of $D$ are 3 and 6. We now use that $H$ has a number of automorphisms: there are 4 pairs of 3-cycles joined by a hamiltonian cycle on their vertices, namely $(abca,xyzx)$, $(ayza,bcxb)$, $(abza, cxyc)$, $(ayca,bzxb)$. Hence up to automorphisms there is only one hamiltonian cycle, namely $C_1=abcxyza$. It is easy to check that \AY{$D \setminus A(C_1)$} is not connected. Hence, if $T,S$ exist then we must have $|A(S)|=7$ and $S$ must consist of a 3-cycle $C$ and a $(C,C)$-path $P$ which picks up the remaining 3 vertices of $V$. By the symmetries above, we may assume that $C=abca$. Again, by permuting the vertices $a,b,c$ if necessary, we can assume that $P$ starts with the arc $cx$. This implies that $P=cxyza$ (as $P$ picks up all the vertices $x,y,z$). Now we see that $S$ contains the hamiltonian cycle $abcxyza$ and we saw above that removing the arcs of this cycle we disconnect the graph.
\end{proof}

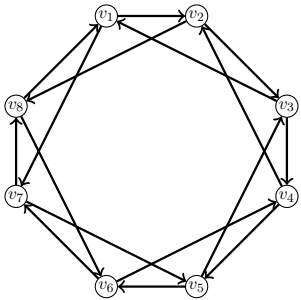
\begin{figure}[H]
\begin{center}
\tikzstyle{vertexB}=[circle,draw, minimum size=14pt, scale=0.6, inner sep=0.5pt]
\tikzstyle{vertexR}=[circle,draw, color=red!100, minimum size=14pt, scale=0.6, inner sep=0.5pt]

\begin{tikzpicture}[scale=0.6]
  \node (v1) at (2,6) [vertexB] {$v_1$};
  \node (v2) at (4,6) [vertexB] {$v_2$};
  \node (v3) at (6,4) [vertexB] {$v_3$};
  \node (v4) at (6,2) [vertexB] {$v_4$};
  \node (v5) at (4,0) [vertexB] {$v_5$};
  \node (v6) at (2,0) [vertexB] {$v_6$};
  \node (v7) at (0,2) [vertexB] {$v_7$};
  \node (v8) at (0,4) [vertexB] {$v_8$};
  \draw[->, line width=0.03cm] (v1) to (v2);
  \draw[->, line width=0.03cm] (v2) to (v3);
  \draw[->, line width=0.03cm] (v3) to (v4);
  \draw[->, line width=0.03cm] (v4) to (v5);
   \draw[->, line width=0.03cm] (v5) to (v6);
  \draw[->, line width=0.03cm] (v6) to (v7);
  \draw[->, line width=0.03cm] (v7) to (v8);
  \draw[->, line width=0.03cm] (v8) to (v1);
   \draw[->, line width=0.03cm] (v1) to (v7);
  \draw[->, line width=0.03cm] (v7) to (v5);
  \draw[->, line width=0.03cm] (v5) to (v3);
  \draw[->, line width=0.03cm] (v3) to (v1);
   \draw[->, line width=0.03cm] (v2) to (v8);
  \draw[->, line width=0.03cm] (v8) to (v6);
  \draw[->, line width=0.03cm] (v6) to (v4);
  \draw[->, line width=0.03cm] (v4) to (v2);

  \end{tikzpicture}
  \caption{A 2-arc-strong  digraph $\hat{D}$ with $\alpha{}(\hat{D})=2$ in which every spanning tree is
    separating.}\label{fig:nonsepT8}
\end{center}
\end{figure}

\begin{proposition}\label{prop:hatD}
  The digraph $\hat{D}$ in Figure \ref{fig:nonsepT8} has no non-separating spanning tree.
\end{proposition}

\begin{proof}
  As every vertex of $\hat{D}$ has in- and out-degree 2 we see that if $T,S$ is a pair of arc-disjoint spanning subdigraphs of $\hat{D}$ such that $T$ is connected and $S$ is strongly connected, then $T$ must be a hamiltonian path in $\hat{D}$ (as $d_S^+(x),d_S^-(x) \geq 1$ for every vertex $x \in V(\hat{D})$
and therefore $d_T^+(x),d_T^-(x) \leq 1$). Let $T=p_1p_2\ldots p_8$.
Let $C$ denote the arcs on the hamiltonian cycle, $v_1v_2\ldots{}v_8v_1$ and let $\bar{C} = A(\hat{D}) \setminus C$.
We first prove the following statement.

\begin{description}
  \item[(i)] $p_i p_{i+1}, p_{i+1} p_{i+2} \in C$ is not possible for any $i \in [6]$.  

For the sake of contradiction, assume the above is true 
and without loss of generality that $p_i = v_1$. That is, $v_1 v_2 v_3$ is a subpath of $T$. Continuing along $T$ out of $v_3$ and into $v_1$ we note that
all arcs of $T$ belong to $C$ (i.e. we cannot use the arc $v_3 v_1$ in $T$, so the only possible arc out of $v_3$ is $v_3 v_4$ and the only possible arc
into $v_1$ is $v_8 v_1$, etc.). 
However $S$ would then contain two disjoint $4$-cycles plus an extra arc, meaning it is not strongly connected, a  contradiction.
\end{description}

As $T$ is a hamiltonian  path we note that $p_i p_{i+1} \in C$ for some $i \in \{2,3,4,5\}$ (otherwise $T$ contains a $4$-cycle).
By (i) we must have $p_{i-1} p_{i}, p_{i+1} p_{i+2} \in \bar{C}$. Without loss of generality assume that  $p_{i}=v_4$,
which implies that $v_6 v_4 v_5 v_3$ is a subpath of $T$. As $T$ is a path (and $i \leq 5$) we must have that $v_6 v_4 v_5 v_3 v_1$ is a subpath of $T$
(as $v_3 v_4 \not\in A(T)$).
This implies that the arcs $v_2 v_3, v_3 v_4, v_4 v_2 \in A(S)$ (as $d_S^+(v_3),d_S^-(v_3),d_S^+(v_4),d_S^-(v_4) \geq 1$).
As $S$ has to contain arcs into and out of $\{v_2,v_3,v_4\}$ we must have $ v_1v_2, v_2v_8 \in A(S)$. But now all arcs incident with $v_2$ are in $S$, a 
contradiction.
\end{proof}

By Proposition \ref{prop:hatD}, the following Conjecture would be best possible in terms of the number of vertices.

\begin{conjecture}
  \label{conj:nonsepT}
Every digraph $D$ on at least 9 vertices with $\lambda{}(D)=2$ and
$\alpha{}(D)\leq 2$ has a non-separating spanning tree
\end{conjecture}

We provide below some support to Conjecture~\ref{conj:nonsepT}, by
proving it for hamiltonian oriented graphs (ie. with no 2-cycle).

\begin{theorem}
\label{thm:orientHC}
Every hamiltonian oriented graph $D=(V,A)$ on at least 9 vertices with
$\lambda{}(D) \AY{\geq 2}$ and $\alpha{}(D)=2$ has a non-separating spanning
tree
\end{theorem}
\begin{proof} 
Let $C$ be a hamiltonian cycle of $D$ and let $X_1,X_2,\ldots{},X_k$
be the vertex sets of the connected components of $H=UG(D)\setminus
A(C)$. If $k=1$ we are done, so assume that $k\geq 2$. Note that each
component has at least 3 vertices as $d_H(v)\geq 2$, \AY{and $D$ 
contains no} 2-cycle. 

Suppose first that $k\geq 3$ and consider a
vertex $v\in X_i$ for some $i\in [k]$. In $UG(D)$ the vertex $v$ has
at most 2 neighbours outside $X_i$. If $v$ has no neighbours in some
$X_q$, $q\neq i$, then let $w$ be a non-neighbour of $v$ in $X_j$,
$j\not\in \{i,q\}$ and let $z$ be an arbitrary non-neighbour of $w$ in
$X_q$. Then $\{v,w,z\}$ is an independent set, contradiction. 
\AY{Therefore $k=3$ and}
every vertex in $X_i$ has a neighbour in each of the other
sets $X_j$, implying that we can pick \AY{$x_i\in X_i$,} $i\in [3]$ such
that $\{x_1,x_2,x_3\}$ is an independent set, contradiction. 

Hence $k=2$ and we may assume that $|X_1|\geq |X_2|$. As $|V|\geq 9$ this
implies that $|X_1|\geq 5$ and hence, as every arc between $X_1$ and
$X_2$ belongs to $C$ this implies that $X_2$ induces a complete
subgraph of $UG(D)$ (every vertex of $X_1$ must be adjacent at least
one of the vertices $x,y$ in a pair of non-adjacent vertices \JBJ{$x,y\in
X_2$ so if such a pair existed, at least one of $x,y$ would be incident to 3 arcs of $C$, contradiction}). Hence, by Theorem~\ref{thm:nonsep14}, we can assume that
$|X_2|\leq 4$ and that $D[X_1]$ is not semicomplete. Note that for
every pair of vertices $u,v\in X_1$ such that $D$ has no arc between
these, every vertex of $X_2$ has at least one edge to $\{u,v\}$ in
$A(C)$. Also note that $\delta^0(D[X_i])\geq 1$, $i\in [2]$ as only
the arcs of $C$ go between $X_1$ and $X_2$.

First suppose that $|X_2|=4$. If $X_1$ contains vertices
$u_1,u_2,v_1,v_2$ all distinct except possibly $u_2=v_1$ so that there
is no arc in $D$ between $u_i$ and $v_i$ for $i=1,2$, then we get the
contradiction that the undirected graph induced by $A(C)$ contains
either a 4-cycle or an 8-cycle, contradicting the $C$ is a hamiltonian
cycle of $D$. Thus $X_1$ contains exactly one pair $u,v$ of
non-adjacent vertices and, by Theorem~\ref{thm:nonsep14}, we may assume that
$|X_1|=5$ so $D$ has 9 vertices.  \jbj{As all 4 vertices of $X_2$ are
  adjacent to either $u$ or $v$, exactly two of them are adjacent to
  $u$ and the other two are adjacent to $v$. Now it is easy to see
  that $C$ has at most one arc inside $X_2$ (otherwise $C$ would
  contain a 3-cycle or a 6-cycle as a subdigraph). Consider first the
  case when $C$ uses no arc inside $X_2$. Then} we can label the
vertices of $V$ such that $X_1=\{v_1,v_2,v_4,v_6,v_8\}$,
$X_2=\{v_3,v_5,v_7,v_9\}$ and $C=v_1v_2\ldots{}v_9v_1$.

\begin{figure}[H]
\begin{center}
\tikzstyle{vertexB}=[circle,draw, minimum size=14pt, scale=0.6, inner sep=0.5pt]
\tikzstyle{vertexR}=[circle,draw, color=red!100, minimum size=14pt, scale=0.6, inner sep=0.5pt]

\begin{tikzpicture}[scale=0.6]]
  \node (v1) at (0,8) [vertexB] {$v_1$};
  \node (v2) at (0,6) [vertexB] {$v_2$};
  \node (v3) at (4,7) [vertexB] {$v_3$};
  \node (v4) at (0,4) [vertexB] {$v_4$};
  \node (v5) at (4,5) [vertexB] {$v_5$};
  \node (v6) at (0,2) [vertexB] {$v_6$};
  \node (v7) at (4,3) [vertexB] {$v_7$};
  \node (v8) at (0,0) [vertexB] {$v_8$};
  \node (v9) at (4,1) [vertexB] {$v_9$};
  \draw[->,line width=0.03cm,color=blue] (v1) to (v2);
  \draw[->,line width=0.03cm,color=blue] (v2) to (v3);
  \draw[line width=0.03cm,color=red] (v3) to (v4);
  \draw[->,line width=0.03cm,color=blue] (v3) to (v5);
  \draw[->,line width=0.03cm,color=blue] (v4) to (v5);
  \draw[->,line width=0.03cm,color=blue] (v5) to (v6);
  \draw[->,line width=0.03cm,color=blue] (v6) to (v7);
  \draw[->,line width=0.03cm,color=blue] (v7) to (v8);
  \draw[->,line width=0.03cm,color=blue] (v8) to (v9);
  \draw[->,line width=0.03cm,color=blue] (v9) to (v1);
  \draw[->,line width=0.03cm,color=blue] (-0.5,4) to (v4);

  \draw (-1,-0.5) rectangle (1,8.5);
  \draw (3,-0.5) rectangle (5,8.5);
  \node () at (0,-1) {$X_1$};
  \node () at (4,-1) {$X_2$};
  \node () at (2,-2) {(a)};

  \draw[line width=0.03cm,color=red] (v5) to (v7);
  \draw[line width=0.03cm,color=red] (v7) to (v9);
  \draw[line width=0.03cm,color=red] (v9) to [out=70, in=-70] (v3);
  
  \node (w1) at (8,8) [vertexB] {$v_1$};
  \node (w2) at (8,6) [vertexB] {$v_2$};
  \node (w3) at (12,7) [vertexB] {$v_3$};
  \node (w4) at (8,4) [vertexB] {$v_4$};
  \node (w5) at (12,5) [vertexB] {$v_5$};
  \node (w6) at (8,2) [vertexB] {$v_6$};
  \node (w7) at (12,3) [vertexB] {$v_7$};
  \node (w8) at (8,0) [vertexB] {$v_8$};
  \node (w9) at (12,1) [vertexB] {$v_9$};
  \draw[->,line width=0.03cm,color=blue] (w1) to (w2);
  \draw[line width=0.03cm,color=red] (w2) to (w3);
  \draw[->,line width=0.03cm,color=blue] (w3) to (w4);
  \draw[->,line width=0.03cm,color=blue] (w4) to (w5);
  \draw[->,line width=0.03cm,color=blue] (w5) to (w6);
  \draw[->,line width=0.03cm,color=blue] (w6) to (w7);
  \draw[->,line width=0.03cm,color=blue] (w7) to (w8);
  \draw[->,line width=0.03cm,color=blue] (w8) to (w9);
  \draw[->,line width=0.03cm,color=blue] (w9) to (w1);

  \draw[->,line width=0.03cm,color=blue] (w2) to (8,5);

  \draw[->,line width=0.03cm,color=blue] (w5) to (w3);
  \draw[->,line width=0.03cm,color=blue] (w7) to (w5);
  \draw[->,line width=0.03cm,color=blue] (w9) to (w7);

  \draw[line width=0.03cm, color=red] (w3) to [out=-70,in=70] (w7);
  \draw[line width=0.03cm, color=red] (w5) to [out=-70,in=70] (w9);
  \draw[line width=0.03cm, color=red] (w3) to [out=-60,in=60] (w9);

  \draw (7,-0.5) rectangle (9,8.5);
  \draw (11,-0.5) rectangle (14,8.5);
  \node () at (8,-1) {$X_1$};
  \node () at (12,-1) {$X_2$};
  \node () at (10,-2) {(b)};

\end{tikzpicture}
\end{center}
\caption{Illustrating two cases in the proof when $|X_2|=4$. In (a) we
  illustrate the solution when $v_3v_5$ is an arc. The blue arcs form
  a strong spanning subdigraph $S$ and the red edges, together with a
  spanning tree $T'$ in $D[X_2]$, avoiding $v_1v_2$ and the blue arc
  \AY{into $v_4$} form a spanning tree $T$ which is edge-disjoint from
  $S$. In (b) we indicate a solution when $D$ contains the directed
  path $v_9v_7v_5v_3$.}\label{fig:X2=4}
\end{figure}

Suppose first that $D$ contains the arc $v_3v_5$. Then let $wv_4$ be
an arbitrary arc entering $v_4$ in $D[X_1]$ (this exists as
$\delta^0(D[X_1])\geq 1$) and let $T'$ be a spanning tree avoiding the
arcs $wv_4,v_1v_2$ in $G'=UG(D[X_1])$. This tree exists as
$G'\setminus \{wv_4,v_1v_2\}$ has 5 vertices and 7 edges and hence is
connected. Then we obtain a strong spanning subdigraph $S$ of $D$ from
$C$ by deleting the arc $v_3v_4$ and adding the arcs $v_3v_5,wv_4$ and
note that $S$ is arc-disjoint from the spanning tree formed by $T'$
and the edges of the path $v_4v_3v_9v_7v_5$ in $UG(D)$, see Figure
\ref{fig:X2=4}(a).

 Hence we can assume that $v_5v_3\in A(D)$ and by
an analogous argument we can assume that $v_7v_5,v_9v_7\in A(D)$. Now
let $v_2w$ be an arbitrary arc leaving $v_2$ in $D[V(X_1)]$, let $T''$
be a spanning tree avoiding the arcs $v_1v_2,v_2w$ in $G'$ and let
$S'$ be the strong spanning subdigraph of $D$ obtained from $C$ by
deleting the arc $v_2v_3$ and adding the arcs of the directed path
$v_9v_7v_5v_3$ and the arc $v_2w$ and note that $S'$ is arc-disjoint
from the spanning tree formed by $T''$ and the edges of the path
$v_5v_9v_3v_7$ (in $UG(D)$) and the arc $v_2v_3$. See Figure
\ref{fig:X2=4}(b).

Next we consider the case when $C$ contains one arc of
  $D[V(X_2)]$. \AY{In this case, we may assume that $v_5$ and $v_8$ are the
two vertices in $X_1$ that are non-adjacent in $D$ and $X_2=\{x_4,x_6,x_7,x_9\}$ 
such that $x_4 x_5 x_6$ and $x_7 x_8 x_9$ are subpaths of $C$.
We may furthermore assume without loss of generality that $x_6 x_7 \in A(C)$ 
(the case when $x_9 x_4 \in A(C)$ is identical, by renaming vertices).
This implies that  we can label $C$ as
  $C=v_1v_2\ldots{}v_9v_1$, and $X_1=\{v_1,v_2,v_3,v_5,v_8\}$ and
 $X_2=\{v_4,v_6,v_7,v_9\}$ (and $v_5$ and $v_8$ are non-adjacent in $D$).}

If  $v_7v_9$ is an arc of $D[V(X_2)]$, then let \AY{$j \in [3]$} be
  chosen such that $v_j$ is an out-neighbour of $v_8$ in
  $D[V(X_1)]$ \AY{and let $j' \in [3]\setminus \{j\}$ be arbitrary.}
  Now the strong spanning subdigraph $S$ consisting of
  the cycle $v_1v_2\ldots{}v_7v_9v_1$ and the path $v_7v_8v_j$ is
  arc-disjoint from the spanning tree using the edges
  $v_4v_6,v_4v_7,v_4v_9,v_9v_8,v_8v_{j'},v_5v_1,v_5v_2,v_5v_3$. Hence
  we can assume that $v_9v_7$ is an arc of $D[V(X_2)]$. A similar
  argument shows that we may assume that $v_6v_4$ is an arc of
  $D[V(X_2)]$. Now using that $\delta^0(D)\geq 2$ this implies that
  the remaining arcs in $D[V(X_2)]$ are $v_7v_4,v_9v_6$ and
  $v_4v_9$. See Figure \ref{fig:oneinX2}(a).

  \begin{figure}[H]
\begin{center}
\tikzstyle{vertexB}=[circle,draw, minimum size=14pt, scale=0.6, inner sep=0.5pt]
\tikzstyle{vertexR}=[circle,draw, color=red!100, minimum size=14pt, scale=0.6, inner sep=0.5pt]

\begin{tikzpicture}[scale=0.6]]
  \node (1) at (3,6) [vertexB] {$v_1$};
  \node (2) at (0,5) [vertexB] {$v_2$};
  \node (3) at (0,-1) [vertexB] {$v_3$};
  \node (4) at (7,-1) [vertexB] {$v_4$};
  \node (5) at (3,0.5) [vertexB] {$v_5$};
  \node (6) at (5,2) [vertexB] {$v_6$};
  \node (7) at (7,4.5) [vertexB] {$v_7$};
  \node (8) at (3,3) [vertexB] {$v_8$};
  \node (9) at (5,6) [vertexB] {$v_9$};
  \draw[->,line width=0.03cm, color=blue] (1) to (2);
  \draw[->,line width=0.03cm, color=blue] (2) to (3);
  \draw[->,line width=0.03cm, color=blue] (3) to (4);
\draw[->,line width=0.03cm, color=blue] (4) to (5);
\draw[->,line width=0.03cm, color=blue] (5) to (6);
\draw[->,line width=0.03cm, color=blue] (6) to (7);
\draw[->,line width=0.03cm, color=blue] (7) to (8);
\draw[->,line width=0.03cm, color=blue] (8) to (9);
\draw[->,line width=0.03cm, color=blue] (9) to (1);
\draw[->,line width=0.03cm] (7) to (4);
\draw[->,line width=0.03cm] (4) to (9);
\draw[->,line width=0.03cm] (9) to (7);
\draw[->,line width=0.03cm] (6) to (4);
\draw[->,line width=0.03cm] (9) to (6);
\draw[line width=0.03cm, dotted] (5) to (8);

\draw (-0.5,-1.5) rectangle (3.5,6.5);
\draw (4.5,-1.5)  rectangle (7.5, 6.5);
\node () at (1.5,-2) {$X_1$};
\node () at (6,-2) {$X_2$};
\node () at (3.7,-3) {(a)};

\node (v1) at (13,6) [vertexB] {$v_1$};
  \node (v2) at (10,5) [vertexB] {$v_2$};
  \node (v3) at (10,-1) [vertexB] {$v_3$};
  \node (v4) at (17,-1) [vertexB] {$v_4$};
  \node (v5) at (13,0.5) [vertexB] {$v_5$};
  \node (v6) at (15,2) [vertexB] {$v_6$};
  \node (v7) at (17,4.5) [vertexB] {$v_7$};
  \node (v8) at (13,3) [vertexB] {$v_8$};
  \node (v9) at (15,6) [vertexB] {$v_9$};
  \draw[line width=0.03cm, dotted] (v5) to (v8);
  \draw[->,line width=0.03cm, color=blue] (v1) to (v2);
  \draw[->,line width=0.03cm, color=blue] (v2) to (v3);
  \draw[->,line width=0.03cm, color=blue] (v3) to (v4);
  \draw[line width=0.03cm, color=red] (v4) to (v5);
  \draw[line width=0.03cm,color=red] (12,1) to (v5);
  \draw[->, line width=0.03cm,color=blue] (12,0) to (v5);
\draw[->,line width=0.03cm, color=blue] (v5) to (v6);
\draw[->,line width=0.03cm, color=blue] (v6) to (v7);
\draw[->,line width=0.03cm, color=blue] (v7) to (v8);
\draw[->,line width=0.03cm, color=blue] (v8) to (v9);
\draw[->,line width=0.03cm, color=blue] (v9) to (v1);
\draw[->,line width=0.03cm,color=blue] (v4) to (v9);
\draw[line width=0.03cm, color=red] (v9) to (v7);
\draw[line width=0.03cm, color=red] (v6) to (v4);
\draw[line width=0.03cm, color=red] (v9) to (v6);

\draw[line width=0.03cm,color=red] (v8) to (v1);
\draw[line width=0.03cm,color=red] (v8) to (v2);
\draw[line width=0.03cm,color=red] (v8) to (v3);

\draw (9.5,-1.5) rectangle (13.5,6.5);
\draw (14.5,-1.5)  rectangle (17.5, 6.5);
\node () at (11.5,-2) {$X_1$};
\node () at (16,-2) {$X_2$};
\node () at (13.7,-3) {(b)};
  
\end{tikzpicture}
\end{center}
\caption{Illustrating the two last cases in the proof when $|X_2|=4$.}\label{fig:oneinX2}
\end{figure}
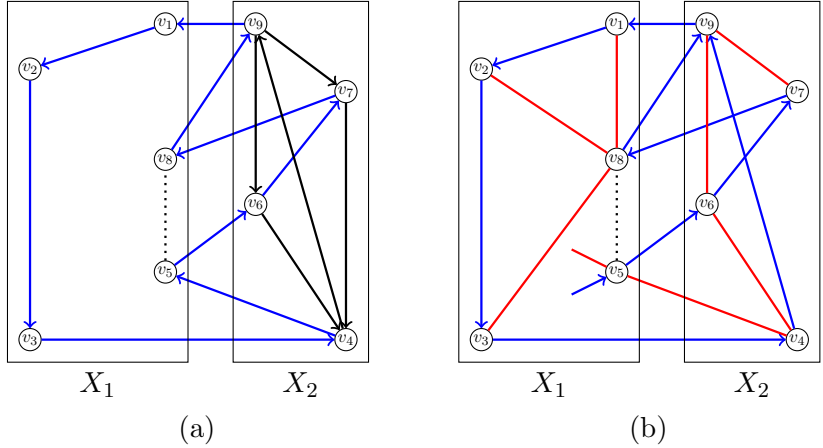

Now choose \AY{$p \in [3]$} such that $v_pv_5$ is an arc of
  $D[V(X_1)]$ and let $S'$ be the strong spanning subdigraph of $D$
  formed by the arcs of the cycle $v_1v_2v_3v_4v_9v_1$ and the path
  $v_pv_5v_6v_7v_8v_9$. 
  \AY{Let $q \in [3] \setminus \{p\}$ be arbitrary and note that $D \setminus A(S')$}
  is connected as it contains the
  spanning tree on the edges $v_4v_5,
  v_4v_6,v_6v_9,v_7v_9,v_8v_1,v_8v_2,v_8v_3,v_5v_q$, see Figure
  \ref{fig:oneinX2}(b). \AY{This completes the case when $|X_2|=4$.}

Consider now the case when $|X_2|=3$. Recall that $X_2$ induces a
3-cycle in $UG(D)$. 
\AY{Let $G$ be the complement of $UG[X_1]$. That is, $V(G)=V(X_1)$ and $uv \in E(G)$ 
if and only if $u$ and $v$ are non-adjacent in $D$.
By Theorem~\ref{thm:nonsep14} we may assume that $\alpha(G) \leq 4$ and as $\alpha(D)=2$
we may assume that $G$ contains no $3$-cycle. As $|V(G)| \geq 6$ we note that we must therefore
have a matching of size two in $G$. Let $uv$ and $u'v'$ be two edges in a matching in $G$.

Note that at least three arcs between $X_2$ and $\{u,v\}$ belong to $C$ (as otherwise there would be an independent set of size 3 
containing $u$ and $v$). There are also at least three arcs between $X_2$ and $\{u',v'\}$ in $C$.
As $|X_2|=3$ these 6 arcs are all the arcs between $X_2$ and $X_1$. Without loss of generality 
assume that $u$ is incident to two arcs between $X_1$ and $X_2$ and $u'$ is also incident to two arcs between $X_1$ and $X_2$,
which implies that both $v$ and $v'$ are incident with exactly one arc between $X_1$ and $X_2$.
As $\alpha(D)=2$ and $\alpha(G)\leq 4$ we now note that $|X_2|=6$ and $E(G)=\{uv,u'v'\}$ or $E(G)=\{uv,u'v',uu'\}$.

Let $u_8=u$, $u_4=v$, $u_6=u'$ and $u_1=v'$. We can now without
loss of generality, label the vertices of $V$ by
$u_1,\ldots{},u_9$ such that $X_1=\{u_1,u_2,u_3,u_4,u_6,u_8\}$,
$X_2=\{u_5,u_7,u_9\}$, $C=u_1u_2\ldots{}u_9u_1$ and there is no arc
between $u_1$ and $u_6$ and no arc between $u_4$ and $u_8$. There may or may not be an arc between $u_6$ and $u_8$. }
 See Figure
\ref{fig:2nonarcsX1}.

\begin{figure}[H]
\begin{center}
\tikzstyle{vertexB}=[circle,draw, minimum size=14pt, scale=0.6, inner sep=0.5pt]
\tikzstyle{vertexR}=[circle,draw, color=red!100, minimum size=14pt, scale=0.6, inner sep=0.5pt]

\begin{tikzpicture}[scale=0.7]]
  \node (u2) at (0,2) [vertexB] {$u_2$};
  \node (u3) at (0,4) [vertexB] {$u_3$};
  \node (u6) at (2,0) [vertexB] {$u_6$};
  \node (u1) at (2,2) [vertexB] {$u_1$};
  \node (u4) at (2,4) [vertexB] {$u_4$};
  \node (u8) at (2,6) [vertexB] {$u_8$};
  \node (u5) at (4,1) [vertexB] {$u_5$};
  \node (u7) at (4,3) [vertexB] {$u_7$};
  \node (u9) at (4,5) [vertexB] {$u_9$};
  \draw[->,line width=0.03cm,color=blue] (u1) to (u2);
  \draw[->,line width=0.03cm,color=blue] (u2) to (u3);
  \draw[->,line width=0.03cm,color=blue] (u3) to (u4);
  \draw[->,line width=0.03cm,color=blue] (u4) to (u5);
  \draw[->,line width=0.03cm,color=blue] (u5) to (u6);
  \draw[->,line width=0.03cm,color=blue] (u6) to (u7);
  \draw[->,line width=0.03cm,color=blue] (u7) to (u8);
  \draw[->,line width=0.03cm,color=blue] (u8) to (u9);
  \draw[->,line width=0.03cm,color=blue] (u9) to (u1);
  \draw[line width=0.03cm, dotted] (u1) to (u6);
  \draw[line width=0.03cm, dotted] (u4) to (u8);
  \draw (-0.5,-0.5) rectangle (2.5,6.5);
  \draw (3.5,-0.5) rectangle (5,6.5);
  \node () at (1,-0.8) {$X_1$};
  \node () at (4,-0.8) {$X_2$};
\end{tikzpicture}
\end{center}
\caption{The hamiltonian cycle $C$ in $D$ when $|X_2|=3$. The dotted
  edges indicate the two pairs of non-adjacent vertices in
  $D[X_1]$.}\label{fig:2nonarcsX1}
\end{figure}
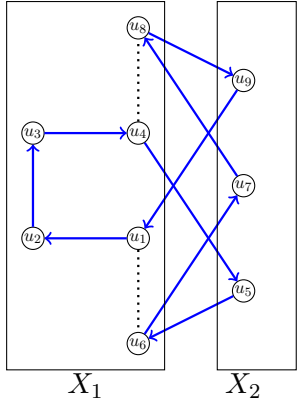

Moreover, as $\delta^+(D)\ge 2$ and $D$ is oriented, we know that
$D[X_2]$ is a directed 3-cycle. If $D$ contains the arc $u_7u_9$, then
as above we can find the desired pair $S,T$, see Figure
\ref{fig:X2=3}(a).

 Otherwise, it means that $D[X_2]$ is the directed
3-cycle $u_9u_7u_5u_9$, and then, as above we can find the desired
pair $S,T$, see Figure \ref{fig:X2=3}(b).\\

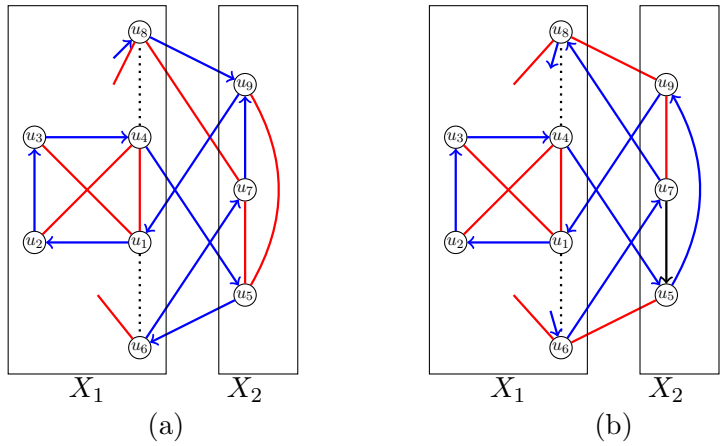
\begin{figure}[H]
\begin{center}
\tikzstyle{vertexB}=[circle,draw, minimum size=14pt, scale=0.6, inner sep=0.5pt]
\tikzstyle{vertexR}=[circle,draw, color=red!100, minimum size=14pt, scale=0.6, inner sep=0.5pt]

\begin{tikzpicture}[scale=0.7]]
  \node (u2) at (0,2) [vertexB] {$u_2$};
  \node (u3) at (0,4) [vertexB] {$u_3$};
  \node (u6) at (2,0) [vertexB] {$u_6$};
  \node (u1) at (2,2) [vertexB] {$u_1$};
  \node (u4) at (2,4) [vertexB] {$u_4$};
  \node (u8) at (2,6) [vertexB] {$u_8$};
  \node (u5) at (4,1) [vertexB] {$u_5$};
  \node (u7) at (4,3) [vertexB] {$u_7$};
  \node (u9) at (4,5) [vertexB] {$u_9$};
  \draw[->,line width=0.03cm,color=blue] (u1) to (u2);
  \draw[->,line width=0.03cm,color=blue] (u2) to (u3);
  \draw[->,line width=0.03cm,color=blue] (u3) to (u4);
  \draw[->,line width=0.03cm,color=blue] (u4) to (u5);
  \draw[->,line width=0.03cm,color=blue] (u5) to (u6);
  \draw[->,line width=0.03cm,color=blue] (u6) to (u7);
  \draw[line width=0.03cm,color=red] (u7) to (u8);
  \draw[->,line width=0.03cm,color=blue] (u8) to (u9);
  \draw[->,line width=0.03cm,color=blue] (u9) to (u1);
  \draw[line width=0.03cm, dotted] (u1) to (u6);
  \draw[line width=0.03cm, dotted] (u4) to (u8);
  \draw[line width=0.03cm, color=red] (u5) to (u7);
  \draw[line width=0.03cm, color=red] (u5) to [out=60,in=-60] (u9);
  \draw[line width=0.03cm, color=red] (u1) to (u3);
  \draw[line width=0.03cm, color=red] (u2) to (u4);
  \draw[line width=0.03cm, color=red] (u1) to (u4);

  \draw[->,line width=0.03cm,color=blue] (1.5,5.5) to (u8);

  \draw[line width=0.03cm, color=red] (u8) to (1.5,5);
  \draw[line width=0.03cm, color=red] (u6) to (1.2,1);

  \draw (-0.5,-0.5) rectangle (2.5,6.5);
  \draw (3.5,-0.5) rectangle (5,6.5);
  \node () at (1,-0.8) {$X_1$};
  \node () at (4,-0.8) {$X_2$};
  \draw[->,line width=0.03cm,color=blue] (u7) to (u9);

  \node (w2) at (8,2) [vertexB] {$u_2$};
  \node (w3) at (8,4) [vertexB] {$u_3$};
  \node (w6) at (10,0) [vertexB] {$u_6$};
  \node (w1) at (10,2) [vertexB] {$u_1$};
  \node (w4) at (10,4) [vertexB] {$u_4$};
  \node (w8) at (10,6) [vertexB] {$u_8$};
  \node (w5) at (12,1) [vertexB] {$u_5$};
  \node (w7) at (12,3) [vertexB] {$u_7$};
  \node (w9) at (12,5) [vertexB] {$u_9$};
  \draw[->,line width=0.03cm,color=blue] (w1) to (w2);
  \draw[->,line width=0.03cm,color=blue] (w2) to (w3);
  \draw[->,line width=0.03cm,color=blue] (w3) to (w4);
  \draw[->,line width=0.03cm,color=blue] (w4) to (w5);
  \draw[line width=0.03cm,color=red] (w5) to (w6);
  \draw[->,line width=0.03cm,color=blue] (w6) to (w7);
  \draw[->,line width=0.03cm,color=blue] (w7) to (w8);
  \draw[line width=0.03cm,color=red] (w8) to (w9);
  \draw[->,line width=0.03cm,color=blue] (w9) to (w1);
  
  \draw[line width=0.03cm, dotted] (w1) to (w6);
  \draw[line width=0.03cm, dotted] (w4) to (w8);
  \draw[->,line width=0.03cm, color=black] (w7) to (w5);
  \draw[->,line width=0.03cm, color=blue] (w5) to [out=60,in=-60] (w9);
  \draw[line width=0.03cm, color=red] (w1) to (w3);
  \draw[line width=0.03cm, color=red] (w2) to (w4);
  \draw[line width=0.03cm, color=red] (w1) to (w4);

  \draw[->,line width=0.03cm,color=blue] (w8) to (9.8,5.3);
  \draw[->,line width=0.03cm,color=blue] (9.8,0.7) to (w6);

  \draw[line width=0.03cm, color=red] (w8) to (9.1,5);
  \draw[line width=0.03cm, color=red] (w6) to (9.1,1);

  \draw (7.5,-0.5) rectangle (10.5,6.5);
  \draw (11.5,-0.5) rectangle (13,6.5);
  \node () at (9,-0.8) {$X_1$};
  \node () at (12,-0.8) {$X_2$};
  \node () at (2.5,-1.5) {(a)};
  \node () at (11,-1.5){(b)};
\draw[line width=0.03cm,color=red] (w9) to (w7);
  
\end{tikzpicture}
\end{center}
\caption{In (a): strong spanning subdigraph (in blue) and spanning
  tree (in red) when $D$ contains the arc $u_7u_5$. In (b): strong
  spanning subdigraph (in blue) and spanning tree (in red) when $D$
  contains the 3-cycle $u_9u_7u_5u_9$}\label{fig:X2=3}
\end{figure}
\end{proof}

\section{Removing a hamiltonian path}\label{sec:removeHP}
Note that Theorem \ref{thm:decomp2asSD} implies that every
2-arc-strong semicomplete digraph $D$ different from $S_4$ has an
out-branching $B^+$ such that \AY{$D \setminus A(B^+)$} is strong. It is easy to
check that $S_4$ also has such an out-branching.  The purpose of this
section is to prove that there exists 2-arc-strong digraphs with
independence number 2 for which no hamiltonian path is
non-separating.\\

For every natural number $r\geq 2$ let $T_r=(V,A)$ be the tournament with vertex set \\
$\{u_0,u_1,\dots{},u_{r+1},v_0,v_1,\dots{},v_{r+1}\}$ and arc set 
$\{u_{i-1}u_i|i\in [r]\}\cup \{v_iv_{i+1}|i\in [r]\}\cup \{u_iv_i|i\in [r]\}\cup \{v_1v_0,v_0u_0,v_0u_1,u_0v_1\}\cup\{u_{r+1}u_r,v_{r+1}u_{r+1},u_rv_{r+1},v_ru_{r+1}\}$ and for all remaining pairs not mentioned above the arcs goes from the vertex of higher index to the one with the lower index. See Figure \ref{fig:Tr}.

\begin{figure}[H]
\begin{center}
  \tikzstyle{vertexB}=[circle,draw, minimum size=18pt, scale=0.7, inner sep=0.5pt]
\tikzstyle{vertexR}=[circle,draw, color=red!100, minimum size=14pt, scale=0.7, inner sep=0.5pt]
\begin{tikzpicture}[scale=1.1]
\node (u0) at (0,2) [vertexB] {$u_0$};
\node (v0) at (0,0) [vertexB] {$v_0$};
\node (u1) at (2,2) [vertexB] {$u_1$};
\node (v1) at (2,0) [vertexB] {$v_1$};
\node (ur) at (10,2) [vertexB] {$u_r$};
\node (vr) at (10,0) [vertexB] {$v_r$};
\node (ur1) at (12,2) [vertexB] {$u_{r+1}$};
\node (vr1) at (12,0) [vertexB] {$v_{r+1}$};
\node (u2)   at (4,2) [vertexB] {$u_2$};
\node (v2)   at (4,0) [vertexB]{$v_2$};
\node (ur-1)   at (8,2) [vertexB]{$u_{r-1}$};
\node (vr-1)   at (8,0) [vertexB]{$v_{r-1}$};
\draw [->, line width=0.03cm] (v0) to (u0);
\draw [->, line width=0.03cm] (v1) to (v0);
\draw [->, line width=0.03cm] (v0) to (u1);
\draw [->, line width=0.03cm] (u0) to (v1);
\draw [->, line width=0.03cm] (u0) to (u1);
\draw [->, line width=0.03cm] (u1) to (u2);
\draw [->, line width=0.03cm] (v1) to (v2);
\draw [->, line width=0.03cm] (ur-1) to (ur);
\draw [->, line width=0.03cm] (vr-1) to (vr);
\draw [->, line width=0.03cm] (ur1) to (ur);
\draw [->, line width=0.03cm] (vr) to (vr1);
\draw [->, line width=0.03cm] (vr) to (ur1);
\draw [->, line width=0.03cm] (ur) to (vr1);
\draw [->, line width=0.03cm] (vr1) to (ur1);
\draw [->, line width=0.03cm, dotted] (u2) to (5,2);
\draw [->, line width=0.03cm, dotted] (v2) to (5,0);
\draw [->, line width=0.03cm, dotted] (7,2) to (ur-1);
\draw [->, line width=0.03cm, dotted] (7,0) to (vr-1);
\draw[->,line width=0.03cm] (u1) to (v1);
\draw[->,line width=0.03cm] (u2) to (v2);
\draw[->,line width=0.03cm] (ur-1) to (vr-1);
\draw[->,line width=0.03cm] (ur) to (vr);

\draw[->, thick,double] (7.5,1) to (4.5,1);

\end{tikzpicture}
\end{center}
\caption{The tournament $T_r$. The fat arc in the middle indicates that all arcs not shown in the figure go from right to left.}\label{fig:Tr}
\end{figure}
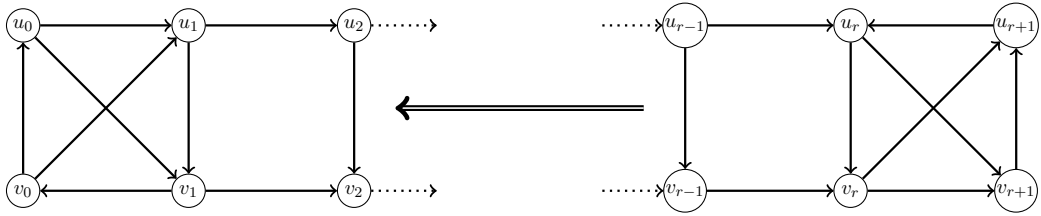

\AY{
\begin{lemma}
  \label{lem:TsepHP}
  For every $r\geq 4$ the tournament $T_r$ is 2-arc-strong.
  Furthermore if $P$ is a hamiltonian  path in $T_r$ starting in $v_0$ then $v_0$ cannot reach $v_r$ in $T_r \setminus A(P)$.
\end{lemma}
}

\begin{proof}
  It is easy to check that $v_0$ has two arc-disjoint paths to every other vertex and that every vertex different from $v_0$ has two arc-disjoint paths to $v_0$. This implies that $T_r$ is 2-arc-strong.
  \AY{For the sake of contradiction assume that there is a hamiltonian  path, $P$, in $T_r$, and that $v_0$ can reach $v_r$ in $S = T_r \setminus A(P)$.
  As for every $i\in [r-1]$ the two arcs $u_iu_{i+1},v_iv_{i+1}$ form a 2-arc-cut of $T_r$ seperating $v_0$ from $v_r$,
one of these arcs must belong to $S$ and the other to $P$.
Similarly, as for every $i\in [r-2]$ the two arcs $u_iu_{i+1},v_{i+1}v_{i+2}$ form a 2-arc-cut of $T_r$ seperating $v_0$ from $v_r$ one of these arcs must belong to $S$ and the other to $P$.
Let $A_1 = \{u_1u_2, u_2u_3, \ldots, u_{r-1} u_r\}$ and let $A_2 = \{v_1v_2, v_2v_3, \ldots, v_{r-1} v_r\}$, and note that by the previous argument 
we must have $A_i \subseteq A(S)$ and $A_{3-i} \subseteq A(P)$ for some $i \in [2]$.
Without loss of generality assume that $A_1 \subseteq A(S)$ and $A_2 \subseteq A(P)$, which implies that $P$ cannot contain both $u_2$ and $u_3$, 
a contradiction to the existence of $P$ and $S$.} 
\end{proof}

\AY{The following corollary follows immediatly from Lemma~\ref{lem:TsepHP}.

\begin{corollary}
  \label{cor:TsepHP}
  For every $r\geq 4$ the tournament $T_r$ is 2-arc-strong and for every hamiltonian path $P$ starting in the vertex $v_0$ the digraph $D\setminus{}A(P)$ is not strongly connected. 
\end{corollary}
}

\begin{theorem}
  \label{thm:noHPok}
  There exist infinitely many 2-arc-strong digraphs $D$ with $\alpha{}(D)=2$
  such that deleting the arcs of any hamiltonian path leaves a non-strong digraph.
\end{theorem}

\begin{proof}
For each \AY{$r\geq 4$} let $T_r$ be the 2-arc-strong tournament defined in
Lemma \ref{lem:TsepHP} and form the digraph $D_r$ from two copies
$T^1_r,T^2_r$ of $T_r$ (whose vertices are superscripted) by adding
two arbitrary arcs from $V(T^i_r)$ to $v^{3-i}_0$ for $i=1,2$. Since
each $T_r$ is a tournament, we have $\alpha{}(D_r)=2$. Moreover,
  as $T^1_r$ and $T^2_r$ arc 2-arc strong, $D_r$ is 2-arc strongly
  connected also. Suppose that $D_r$ has a hamiltonian path $P$ such
that \AY{$D \setminus A(P)$} is strong. Without loss of generality $P$ starts in
$V(T^1_r)$ and thus the restriction of $P$ to $V(T^2_r)$ is a
hamiltonian path $P'$ starting at $v^2_0$. \AY{By 
Lemma \ref{lem:TsepHP} we note that $v^2_0$ cannot reach $v^2_r$ in $T^2_r \setminus A(P')$,
which implies that no vertex in $T^1_r$ can reach $v^2_r$ in $D_r \setminus A(P)$.
So $D_r \setminus A(P)$ is not strong, a contradiction.}
  \end{proof}

\section{Non-separating hamiltonian paths in  graphs
  with independence number 2}\label{sec:nonsepHPG}

In contrast to the result in Theorem \ref{thm:noHPok} above, for the case of undirected graphs of independence number 2 we have the following positive result on non-separating hamiltonian paths.

\begin{theorem}
Let $G$ be a 2-edge-connected graph with $\delta(G)\ge 4$ and
$\alpha(G)\le 2$. Then, $G$ contains a spanning tree and a Hamiltonian
path which are edge-disjoint.
\end{theorem}

\begin{proof}
Let $G$ be a 2-edge connected graph with $\delta(G)\ge 4$ and
$\alpha(G) \le 2$.  It is easy to see that every connected graph with
independence number at least 2 has a spanning tree with a number of
leaves at most its independence number. Hence $G$ contains a
Hamiltonian path $P$.  If AY{$G \setminus E(P)$} is connected, we are done. Otherwise
let $X_1,X_2,\dots ,X_p$ be the connected components of \AY{$G \setminus E(P)$.}  Notice
that as $\delta(G)\ge 4$ and $P$ is a path, we have $\delta(\AY{G \setminus E(P)})\ge
2$, and in particular, we have $|X_i|\ge 3$ for all $i=1,\dots ,p$. If
$p\ge 3$, consider $u$ an extremity of $P$ and assume without loss of
generality that $u\in X_1$ and that its neighbour $v$ in $P$ is in
$X_1\cup X_2$. It means that all the vertices of $X_2\setminus \{v\}$
and $X_3$ are non neighbours of $u$ and hence must form a complete
subgraph of $G$. In particular, all the edges between $X_2\setminus
\{v\}$ and $X_3$ are edges of $P$, which is not possible as
$|X_2\setminus \{v\}|\ge 2$ and $|X_3|\ge 3$, implying that $P$ would
contain a cycle. So, \AY{$G \setminus E(P)$} contains exactly the two connected
components $X_1$ and $X_2$.

Notice that the case $|X_1|=|X_2|=3$ is not possible, as in this case,
as $\delta(G)\ge 4$, every vertex of $X_1$ would have at least two
neighbours in $X_2$ and every vertex of $X_2$ would have at least two
neighbours in $X_1$, and so $P$ would contain a cycle, which is not
possible. Hence $\max\{|X_1|,|X_2|\}\geq 4$ and we may assume that
$|X_1|\geq 4$.

Assume first that $\alpha(G[X_1])=\alpha(G[X_2])=1$, that is, they are
both complete graphs. In this case, we will show how to build a
Hamiltonian path and a spanning tree of $G$ which are
edge-disjoint. If $x$ is a vertex of any complete graph $K$ on at
least 4 vertices, it is easy to find a Hamiltonian path which starts
in $x$ and a spanning tree which are edge-disjoint. Suppose first that
$|X_2|\ge 4$. Then it follows from the fact that $G$ is 2-edge
connected, that there exist two distinct edges $x_1x_2$ and $y_1y_2$
with $x_1,y_1\in X_1$ and $x_2,y_2\in X_2$. Now consider for $i=1,2$ a
Hamiltonian path $P_i$ of $G[X_i]$ starting in $x_i$ and a spanning
tree $T_i$ of $G[X_i]$ edge-disjoint from $P_i$. We conclude by
considering the Hamiltonian path $(P_1\cup P_2)+ x_1x_2$ of $G$
edge-disjoint from the spanning tree $(T_1\cup T_2)+ y_1y_2$ of
$G$. Hence we may assume that $|X_2|=3$ and denote the vertices in
$X_2$ by $\{x_2,y_2,z_2\}$.  As $\delta(G)\ge 4$, there exist \AY{three
distinct edges of $G$ $x_2x_1,y_2y_1,z_2z_1$ such that
$x_1$, $y_1$} and $z_1$ belong to $X_1$. So, we consider a
Hamiltonian path $P_1$ of $G[X_1]$ starting in \AY{$x_1$} and a spanning
tree $T_1$ of $G[X_1]$ edge-disjoint from $P_1$. We conclude then with
the Hamiltonian path \AY{$P'=(P_1 \cup x_2y_2z_2)+x_2x_1$ of $G$ and the
spanning tree $T_1 +x_2z_2+y_2y_1+z_2z_1$} of
$G$ edge-disjoint from $P$.\\

\AY{If $\alpha(G[X_1])=1$ and $\alpha(G[X_2])=2$, then as $\delta(G) \geq 4$, 
we must have $|X_2| \geq 4$.  In this case we may swap the names of $X_1$ and $X_2$,
which implies that we may assume without loss of generality that $\alpha(G[X_1])=2$. Let} 
$x_1$ and $y_1$ be two  vertices of $X_1$ which are not adjacent in
$G$. As every vertex of $X_2$ is adjacent to $x_1$ or $y_1$ in $G$ and
as the corresponding edges must be edges of $P$, we have $|X_2|\le 4$.  Suppose 
$|X_2|=4$. Then we will   prove that $G[X_1]$ is almost complete, that is it
contains all the possible edges except $x_1y_1$. Denote the vertices of $X_2$ by
$\{x_2,y_2,z_2,t_2\}$, and as $\alpha(G)=2$, we can assume that $x_2$
and $y_2$ are adjacent to $x_1$ and that $z_2$ and $t_2$ are adjacent
to $y_1$. Assume first that $G[X_1]$ contains two non-adjacent vertices
$z_1$ and $t_1$ both distinct from $x_1$ and $y_1$. As
$\alpha(G)=2$, the vertex $x_2$ has to be adjacent to $z_1$ or $t_1$.
We can assume that $z_1x_2$ is an edge of $G$, but then $y_2$ has to
be adjacent to $t_1$ as there is no cycle induced by the edges of
$P$. Similarly, we can assume that $z_1z_2$ is an edge of $P$, but
then $t_2$ cannot be adjacent to any of $z_1,t_1$ without
creating a cycle induced by the edges of $P$. So, $\{z_1,t_1,t_2\}$ is
an independent set, contradicting $\alpha(G)=2$. 
Now  assume
that $X_1$ contains a vertex $z_1\neq y_1$ which is non-adjacent to $x_1$.
Then $z_2$
has to be adjacent to $x_1$ or $z_1$ and as $x_2$ and $y_2$ are already
adjacent to $x_1$ in $P$, $z_2$ must be adjacent to $z_1$. Similarly,
$t_2$ is adjacent to $z_1$, but then $z_1z_2y_1t_2$ would form a cycle
with the edges of $P$, a contradiction. Similarly we can prove  that every
vertex of $X_1$ except $x_1$ is adjacent to $y_1$ so $G[X_1]$
is the graph $K_{|X_1|}-x_1y_1$.
As $|X_1|\ge 4$, it is easy to find
then in $G[X_1]$ two vertex-disjoint paths $P_1$ and $P'_1$ such that
$V(P_1)\cup V(P'_1)=X_1$, the path $P_1$ ends in $x_1$ and the path
$P'_1$ ends in $y_1$ and $G[X_1]-P_1-P_2$ is connected and so contains
a spanning tree $T_1$. On the other hand, as $\delta(G[X_2]-P)\ge 2$,
the graph $G[X_2]$ contains a 4-cycle. One of the edges of this 4-cycle
goes from $\{x_2,y_2\}$ to $\{z_2,t_2\}$. So, denote this 4-cycle by
$abcda$ such that $a$ is adjacent to $x_1$ and $d$ is adjacent to
$y_1$. Now, we consider the Hamiltonian path $P'=(P_1\cup
P'_1)+ab+bc+cd$ of $G$. As $b$ and $c$ have at least one neighbour
each in $X_1$, we know that \AY{$G \setminus E(P')$} has at most 2 connected components,
one containing $X_1\cup \{b,c\}$, the other one containing
$\{a,d\}$. But as $\delta(G)\ge 4$, $a$ has a neighbour in
$X_1\cup \{b,c\}$ different from $x_1$, and so, \AY{$G \setminus E(P')$} is
connected.\\

\AY{The only remaining case is when $|X_2|=3$.} Denote the vertices of $X_2$ by $\{x_2,y_2,z_2\}$ and assume without loss
of generality that $x_2$ and $y_2$ are adjacent to $x_1$ and that
$z_2$ is adjacent to $y_1$. As $z_2$ as degree at least 4 in $G$,
there exist a vertex $z_1\in X_1$ distinct from $x_1$ and $y_1$ such
that $z_1$ is a neighbour of $z_2$ in $P$.  More generally, as
$\delta(G)\ge 4$, every vertex of $X_2$ has exactly two neighbours in
$X_1$, which are then neighbours in $P$, $G[X_2]$ is a complete graph
and finally every vertex of $X_2$ has degree exactly 4. In particular,
$\{x_1,z_2\}$ is an independent set of size 2, and every vertex of
$X_1\setminus \{x_1,y_1,z_1\}$ is adjacent to $x_1$. Now, let us focus
on the extremities of the path $P$. Both cannot lie in
$Y=\{x_1,y_1,z_1,x_2,y_2,z_2\}$ as the only vertices of $Y$ with
degree less than 2 in $P$ are $y_1$ and $z_1$, and they cannot be both
extremities of $P$ as otherwise, $P$ would be the path
$y_1z_2z_1$. So, denote by $p$ an extremity of $P$ not lying in $Y$ and recall that  $x_1$ is adjacent to every vertex of 
$V\setminus Y$ so $x_1p$ is an edge of $G$.
Now consider the Hamiltonian path $P'$ of $G$ defined by
$P'=(P-x_2x_1-x_1y_2)+x_2y_2+x_1p$.  To conclude,
let us prove that \AY{$G \setminus E(P')$} is connected. Indeed, every vertex of
$V\setminus \{y_1,z_1,z_2\}$ is linked to $x_1$, and all the
corresponding edges except $px_1$ are edges of \AY{$G \setminus E(P')$}. So, \AY{$G \setminus E(P')$}
induces a connected graph on $Z=V\setminus \{y_1,z_1,z_2,p\}$.
Moreover, $z_2$ is adjacent to $x_2$ and $x_2z_2$ is not an edge of
$P'$. And by choice, $p$ is not adjacent to $z_2$ and so has a
neighbour in $Z$ different from $x_1$. Finally, $y_1$ and $z_1$ have
both at least one neighbour in \AY{$G \setminus E(P')$} which belongs to $V\setminus
\{y_1,z_1\}$. Thus, \AY{$G \setminus E(P')$} is connected and the proof is complete.
\end{proof}

Notice that we cannot replace $\delta{}(G)\geq 4$  by $\delta(G)=3$ (even if $\lambda(G)=3$) as
shown by the graph built from two 3-cycles linked by a perfect
matching.
\AY{Also, any $3$-regular graph, $G$, has $|E(G)|=3|V(G)|/2$, so cannot contain two edge-disjoint spanning
trees when $|V(G)|>4$, and therefore also not a hamiltonian path and a spanning tree that are edge-disjoint.}

\section{Remarks and open problems}\label{sec:remarks}

All proofs in this paper are constructive and it is not difficult to derive polynomial algorithms for finding the desired objects in case they exist. We leave the details to the interested reader.

\begin{problem}
  Determine the complexity of deciding whether a strong digraph of independence number 2 has a non separating out-branching.
\end{problem}

\JBJ{\begin{problem}
  Determine the complexity of deciding whether a strong digraph of independence number 2 has a non separating spanning tree.
\end{problem}

This problem is NP-complete for general digraphs as shown in \cite{bangTCS438}.\\}

Theorem \ref{thm:nonsep14} suggests that perhaps we can get rid of the requirement on the minimum \JBJ{in-degree} in Theorem \ref{thm:main} when the digraph has enough vertices.

\begin{conjecture}\label{conj:nonsepnoroot}
  There exists an integer $K$ such that every digraph $D$ on at least $K$ vertices with $\lambda{}(D)\geq 2$ and $\alpha{}(D)=2$ has a non-separating out-branching.
\end{conjecture}

It is not difficult to check that every member of the infinite class of digraphs that we used in Proposition \ref{prop:no2colstrong}
has a non-separating branching from every vertex.

\begin{conjecture}\label{conj:nonseproot}
  There exists an integer $L$ such that every digraph $D$ on at least $L$ vertices  with $\lambda{}(D)\geq 2$ and $\alpha{}(D)=2$ has a non-separating out-branching $B^+_s$ for every choice of $s\in V$.
\end{conjecture}

\begin{question}
  Does every 3-arc-strong digraph $D$ with $\alpha{}(D)=2$ have a pair of arc-disjoint spanning strong subdigraphs?
\end{question}

In Proposition \ref{prop:no2colstrong} we showed that 2-arc-strong connectivity and high minimum semi-degree is not enough to guarantee such digraphs.

\bibliography{refs}

\end{document}